\documentclass[12pt]{article}
\usepackage[dvips]{graphics,color}
\usepackage{latexsym}
\usepackage{amsmath}
\usepackage{amssymb}
\usepackage{amsfonts,amsthm}

\newtheorem{theorem}{Theorem}[section]
\newtheorem{lemma}[theorem]{Lemma}

\newtheorem{corollary}[theorem]{Corollary}

\newenvironment{remark}{\noindent{\bf Remark:}}{}

\newcommand{\ignore}[1]{}

\begin{document}

\title{Finding Dominators via Disjoint Set Union}

\date{\today}

\author{Wojciech Fraczak$^{1}$ \and Loukas Georgiadis$^{2}$ \and Andrew Miller$^{3}$ \and Robert E. Tarjan$^{4}$}

\maketitle

\begin{abstract}
The problem of finding dominators in a directed graph has many important applications, notably in global optimization of computer code.  Although linear and near-linear-time algorithms exist, they use sophisticated data structures.  We develop an algorithm for finding dominators that uses only a ``static tree'' disjoint set data structure in addition to simple lists and maps. The algorithm runs in near-linear or linear time, depending on the implementation of the disjoint set data structure.  We give several versions of the algorithm, including one that computes loop nesting information (needed in many kinds of global code optimization) and that can be made self-certifying, so that the correctness of the computed dominators is very easy to verify.
\end{abstract}

\footnotetext[1]{Universit\'{e} du Qu\'{e}bec en Outaouais, Gatineau, Qu\'{e}bec, Canada.}
\footnotetext[2]{Department of Computer Science \& Engineering, University of Ioannina, Greece. E-mail: \texttt{loukas@cs.uoi.gr}.}
\footnotetext[3]{Benbria Corporation, Ottawa, Ontario, Canada.}
\footnotetext[4]{Department of Computer Science, Princeton University, Princeton, NJ, and Microsoft Research Silicon Valley. E-mail: \texttt{ret@cs.princeton.edu}. Research at Princeton University partially supported by NSF grant CCF-0832797.}

\section{Introduction}
\label{sec:introduction}

A \emph{flow graph} $G = (V, A, s)$ is a directed graph with vertex set $V$, arc set $A$, and a distinguished \emph{start vertex} $s$ such that every vertex is reachable from $s$.  A vertex $u$ \emph{dominates} another vertex $v$ in a flow graph $G$ if every path from $s$ to $v$ contains $u$.  The dominator relation is reflexive and transitive.  Furthermore, the graph of the transitive reduction of this relation is a tree rooted at $s$, called the \emph{dominator tree} $D$: every vertex $v \not = s$ has an \emph{immediate dominator} $d(v) \not= v$, the parent of $v$ in $D$, such that all dominators of $v$ other than $v$ also dominate $d(v)$. Thus $D$, which we represent by its parent function $d$, succinctly represents the dominator relation. Our goal is to find the dominator tree of a given flow graph $G$. See Figure \ref{fig:dominator-tree}.

\begin{figure}[t]
\begin{center}
\scalebox{0.7}[0.7]{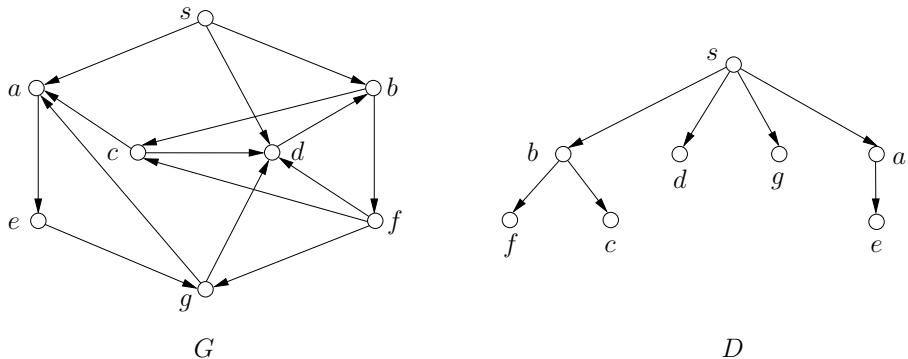}
\end{center}
\vspace{-0.5cm}
\caption{\label{fig:dominator-tree} A flow graph and its dominator tree.}
\end{figure}

Dominators have a variety of important applications, notably in optimizing compilers \cite{pcd:au,cytron:91:toplas} but also in many other areas \cite{foodwebs:ab04,amyeen:01:vlsitest,2vc,2VCSS:Geo,Rodrigues:icml12,Italiano2012,memory-leaks:mgr2010,QVDR:PADL:2006}.  Thus it is not surprising that considerable effort has gone into devising fast algorithms for finding dominators. In 1979, Lengauer and Tarjan \cite{domin:lt} presented an algorithm, the \emph{LT algorithm}, that is fast both in theory and in practice \cite{DomCertExp:SEA13,dom_exp:gtw} and that has been widely used: Lengauer and Tarjan included pseudo-code for the complete algorithm in their paper. The algorithm has one conceptually complicated part, a data structure for computing minima on paths in a tree \cite{pathcomp:t}. The running time of the LT algorithm depends on the implementation of this data structure: with a simple implementation, the LT algorithm runs in $O(m \log{n})$ time on an $n$-vertex, $m$-arc graph; with a sophisticated implementation, it runs in $O(m \alpha(n, m/n))$ time.  Here $\alpha$ is a functional inverse of Ackermann's function
defined as follows: For natural numbers $k$ and $j$, let $A(k, j)$ be defined recursively by $A(0, j) = j + 1$, $A(k, 0) = A(k - 1, 1)$ if $k > 0$, and $A(k, j) = A(k - 1, A(k, j - 1))$ if $j, k > 0$.  Let $\alpha(n, d) = \min \{ k > 0 \ | \ A(k, \lfloor d \rfloor) > n \}$. Function $\alpha$ grows extremely slowly; it is constant for all practical purposes. Later work produced more-complicated but truly linear-time variants of the LT algorithm \cite{domin:ahlt,dominators:bgkrtw,domin:bkrw,dom:gt04}.

In practice, the simple version of the LT algorithm performs at least as well as the sophisticated version, in spite of the smaller worst-case time bound of the former \cite{DomCertExp:SEA13,dom_exp:gtw}. The complexity of the underlying data structure has led researchers to seek both simpler fast algorithms \cite{dom:chk01,dominators:poset,dom_exp:gtw,loops:Ramalingam} and ways to certify the correctness of the output of a dominator-finding algorithm \cite{domv:gt05,domcert,bipolar:2012,domver:ZZ}. Notable results related to our work are the following.

Ramalingam and Reps~\cite{RR:incdom} gave an algorithm for finding dominators in an acyclic graph by computing nearest common ancestors (nca's) in a tree that grows by leaf additions. For this version of the nca problem, Gabow~\cite{nca:g} gave an $O(m)$-time RAM algorithm, and Alstrup and Thorup~\cite{dnca:at00} gave an $O(m \log{\log{n}})$-time pointer machine algorithm.  These algorithms give implementations of the Ramalingam-Reps algorithm that run in $O(m)$ time on a RAM and $O(m \log{\log{n}})$ time on a pointer machine, respectively.  Later, Ramalingam~\cite{loops:Ramalingam} gave a reduction of the problem of finding dominators in a general graph to the same problem in an acyclic graph.  Although he did not mention it, his reduction is an extension of the algorithm of Tarjan~\cite{st:t} for finding a loop nesting forest. In addition to simple lists and maps, his reduction uses a data structure for the ``static tree'' version of the disjoint set union problem \cite{dsu:gt,dsu:tarjan}.  There are simple, well-known  disjoint-set data structures with an inverse-Ackermann-function amortized time bound per operation \cite{dsu:tarjan,setunion:tvl}. Use of any of these in Ramalingam's reduction results in an $O(m \alpha(n, m/n))$ running time.  Furthermore static tree disjoint set union has an $O(m)$-time RAM algorithm \cite{dsu:gt}.  Use of this in Ramalingam's reduction gives an $O(m)$-time RAM implementation.  Combining Ramalingam's reduction with the Ramalingam-Reps algorithm for acyclic graphs gives an algorithm for finding dominators that runs in $O(m)$ time on a RAM or $O(m \log{\log{n}})$ time on a pointer machine, depending on the implementation.

Georgiadis and Tarjan~\cite{domcert,bipolar:2012} developed methods for making a dominator-finding algorithm \emph{self-certifying}, by adding the computation of a \emph{low-high order}. One of their methods requires only simple data structures and a loop nesting forest, which can be computed by Tarjan's algorithm~\cite{st:t}.

Recently, Gabow \cite{dominators:poset} has developed a dominator-finding algorithm that uses only simple data structures and a data structure for static tree set union, thereby eliminating the need for finding path minima in a tree or computing nca's in an incremental tree. His algorithm is based on his previous work on the minimal-edge poset \cite{poset:Gabow91,minset:poset}.

Our work builds on these results.  We develop a dominator-finding algorithm that, like Gabow's, uses only simple data structures and a data structure for static tree set union. Our algorithm does different computations than Gabow's algorithm, and it differs from his in several other ways. We compare our algorithm with his in Section \ref{sec:final}.

In addition to this introduction, our paper contains four sections.  In Section \ref{sec:acyclic} we develop a dominator-finding algorithm for the special case of acyclic graphs. As part of its initialization, the algorithm finds a spanning tree rooted at $s$ and computes nearest common ancestors (nca's) in this tree.  Any spanning tree will do, but the algorithm becomes simpler if the spanning tree is depth-first and its vertices are processed in reverse preorder.  This eliminates the need to compute nca's.  Section \ref{sec:general} extends the algorithm of Section \ref{sec:acyclic} to general graphs.  The extension requires the spanning tree to be depth-first, and it requires an nca computation.  Section \ref{sec:loops} describes a variant of the algorithm of Section \ref{sec:general} that runs Tarjan's algorithm for finding a loop-nesting forest as part of the computation. This eliminates the need to compute nca's, and it allows the algorithm to be easily extended to make it self-certifying. Section \ref{sec:final} contains final remarks, including a comparison of our algorithm with Gabow's.  We develop our algorithm in as general a way as possible.  This leaves several design decisions up to the implementer, such as whether to keep the several passes of the algorithm separate or to combine them.

Our paper is a completely rewritten and extended version of a conference paper \cite{dominators:FM12} by the first and third authors.  The algorithm in that paper has a gap (discussed in Section \ref{sec:general}) that was corrected by the second and fourth authors.

\section{Finding Dominators in an Acyclic Graph}
\label{sec:acyclic}

In the remainder of our paper, $G = (V, A, s)$ is a flow graph with $n$ vertices and $m$ arcs, $D$ is the dominator tree of $G$, and $d$ is the parent function of $D$.  Tree $D$ has root $s$ and vertex set $V$, but it is not necessarily a spanning tree of $G$, since its arcs need not be in $A$.  To simplify time bounds we assume $n > 1$, which implies $m > 0$ since all vertices are reachable from $s$.  We assume that there are no arcs into $s$; such arcs do not change $D$.  We assume that the original graph contains no multiple arcs (two or more arcs $(v, w)$) and no loop arc (an arc $(v, v)$).  Such arcs can be created by the contractions done by the algorithm, but they do not affect the dominators.  A graph is \emph{acyclic} if it contains no cycle of more than one vertex; thus a loop arc is not a cycle. In an abuse of notation, we denote an arc by the ordered pair of its end vertices: even though there may be several arcs with the same ends, they are all equivalent.  To make graph search efficient, we assume that each vertex $v$ has a set of its outgoing arcs $(v, w)$ and a set of its incoming arcs $(u, v)$, or equivalently a set of the vertices $w$ such that $(v, w)$ is an arc and a set of the vertices $u$  such that $(u, v)$ is an arc.  This allows examination of the arcs out of a vertex or into a vertex in time proportional to their number.

Let $T$ be an arbitrary spanning tree of $G$ rooted at $s$.   For each vertex $v \not= s$, let $p(v)$ be the parent of $v$ in $T$.  Arc $(v, w)$ of $G$ is a \emph{tree arc} if $v = p(w)$, a \emph{forward arc} if $v$ is a proper ancestor of $p(w)$ in $T$ (by \emph{proper} we mean $v \not= p(w)$), a \emph{back arc} if $v$ is a proper descendant of $w$ in $T$, a \emph{cross arc} if $v$ and $w$ are unrelated in $T$, and a \emph{loop arc} if $v = w$. See Figure \ref{fig:spanning-tree}. Note that we allow multiple tree arcs into the same vertex $w$.  Such arcs can be created by contractions.  An alternative definition is to specify one such arc to be the tree arc into $w$ and define the others to be forward arcs.  Our algorithm does not distinguish between tree arcs and forward arcs, so either definition works, as does defining all tree arcs to be forward arcs: we keep track of the current spanning tree via its parent function, not its arcs.  For any vertex $v$, the path in $T$ from $s$ to $v$ avoids all vertices that are not ancestors of $v$ in $T$, so $d(v)$ is a proper ancestor of $v$ in $T$.

\ignore{
\begin{figure}[t]
\begin{center}
\scalebox{0.7}[0.7]{\input{acyclic-example-1.pstex_t}}
\end{center}
\caption{\label{fig:acyclic-1} An acyclic flow graph, its dominator tree, and a spanning tree.}
\end{figure}
}

\begin{figure}[t]
\begin{center}
\scalebox{0.7}[0.7]{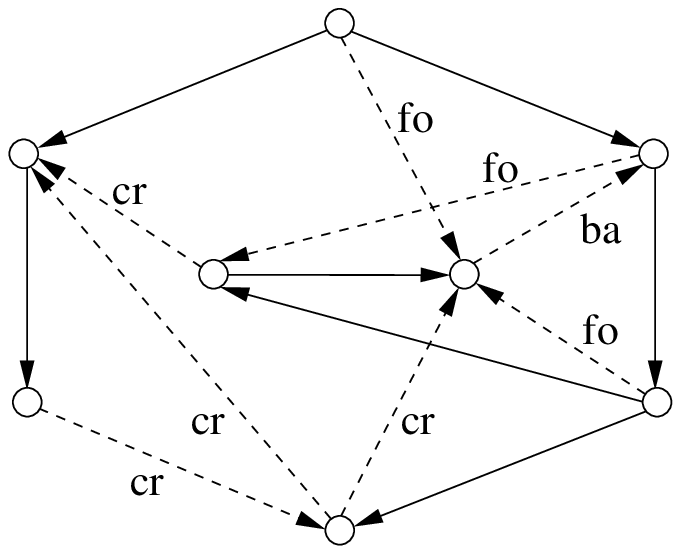}
\end{center}
\vspace{-0.5cm}
\caption{\label{fig:spanning-tree} A spanning tree of the flow graph in Figure \ref{fig:dominator-tree} shown with solid arcs; non-tree arcs are shown dashed and labeled by type: $\mathrm{fo}$ = forward arc, $\mathrm{cr}$ = cross arc, and $\mathrm{ba}$ = back arc.}
\end{figure}

Our dominator-finding algorithm for acyclic graphs, Algorithm AD, begins by building an arbitrary spanning tree $T$ rooted at $s$. It then repeatedly selects a vertex $v \not= s$ and contracts it into its parent in $T$, thereby deleting $v$ from both $G$ and $T$.  Each such contraction preserves the dominators of the undeleted vertices. For each undeleted vertex $v$, the algorithm maintains a set $\mathit{same}(v)$ of vertices having the same immediate dominator as $v$. Just before contracting $v$ into its parent, the algorithm either computes $d(w)$ for each vertex $w$ in $\mathit{same}(v)$, including $v$, or identifies a non-deleted vertex $x$ such that $d(v) = d(x)$ and adds all vertices in $\mathit{same}(v)$ to $\mathit{same}(x)$. Contractions continue until only $s$ remains, by which time all immediate dominators have been computed.

We begin the development of Algorithm AD by discussing contractions and establishing sufficient conditions for a contraction to
\begin{itemize}
\item[(a)] provide a way to compute the immediate dominator of the deleted vertex from those of the undeleted vertices, and
\item[(b)] preserve the dominators of the undeleted vertices.
\end{itemize}
A \emph{contraction} of vertex $v$ \emph{into} vertex $x \not= v$ replaces each arc end equal to $v$ by $x$.  That is, it replaces each arc $(u, v)$ with $u \not= v$ by $(u, x)$, each arc $(v, y)$ with $y \not= v$ by $(x, y)$, and each loop arc $(v, v)$ by $(x, x)$.  Any arc $(x, v)$ or $(v, x)$ becomes a new loop arc $(x, x)$. See Figure \ref{fig:contract}. Each contraction done by our algorithm contracts a vertex $v \not= s$ into its parent $p(v)$ in the current spanning tree $T$.  In addition to changing the graph, the contraction changes $T$ into a new spanning tree $T'$ rooted at $s$, whose arcs are those of $T$ not into $v$.  The parent function $p'$ of $T'$ is defined for $w \not \in \{s, v\}$ and is $p'(w) = p(w)$ if $p(w) \not= v$, $p'(w) = p(v)$ if $p(w) = v$.

\begin{figure}[t]
\begin{center}
\scalebox{0.7}[0.7]{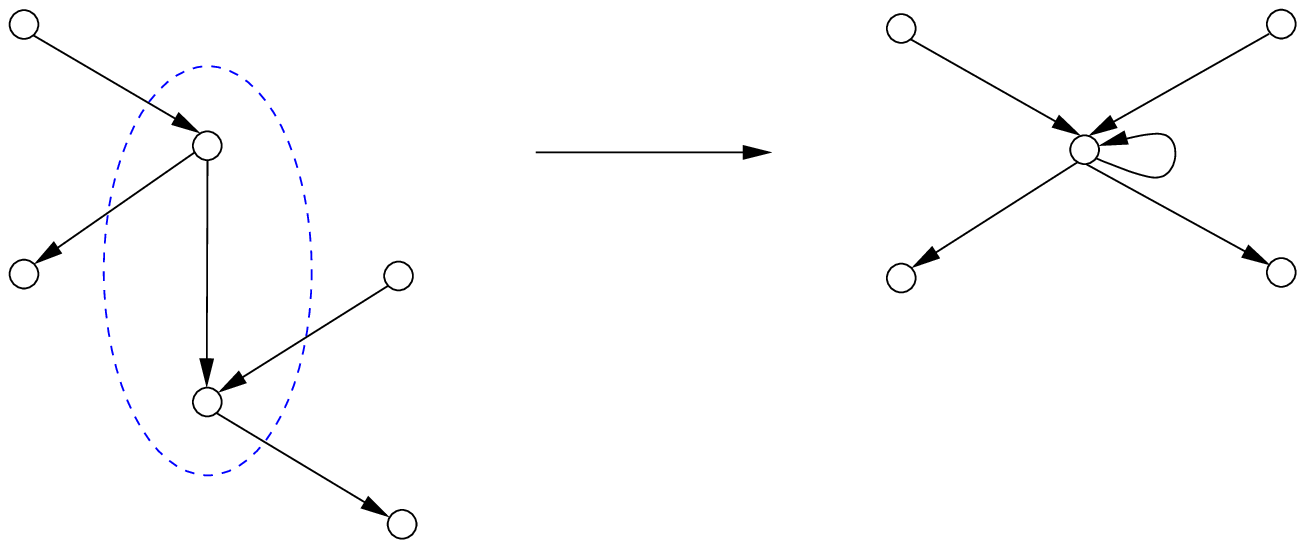}
\end{center}
\vspace{-0.5cm}
\caption{\label{fig:contract} Vertex contraction.}
\end{figure}

\begin{lemma}
\label{lemma:acyclic-contract}
Suppose $G$ is acyclic.  Let $v \not= s$, and let $G'$ be the graph formed from $G$ by contracting $v$ into $p(v)$.  Graph $G'$ contains one new loop arc $(p(v), p(v))$ for each arc $(p(v), v)$ in $G$, and no other new loop arcs.
\end{lemma}
\begin{proof}
Suppose $G$ is acyclic.  Every arc added by the contraction corresponds to a path in $G$, so such additions cannot create any cycles.  Each tree arc $(p(v), v)$ in $G$ becomes a loop arc $(p(v), p(v))$ in $G'$.  Since $G$ is acyclic, it contains no arc $(v, p(v))$, so no such arc becomes a loop arc in $G'$.
\end{proof}

The following two lemmas provide properties (a) and (b) of a contraction, respectively.  They hold for arbitrary flow graphs.  Assign the vertices of $T$ distinct integers from $1$ to $n$ in a bottom-up order (an order such that $v$ is numbered less than $p(v)$ for all $v \not= s$) and identify vertices by number.

\begin{lemma}
\label{lemma:parent-dom}
Let $v \not= s$ be a vertex with no entering cross arc or back arc, let $(u, v)$ be an arc with $u$ maximum, and suppose $d(p(v)) \ge u$.  If $u = p(v)$, then $d(v) = u$; otherwise, $d(v) = d(p(v))$.
\end{lemma}
\begin{proof}
If $u = p(v)$, then every arc into $v$ is from $u$.  (There may be more than one, if there are multiple arcs.)  Thus $u$ dominates $v$.  Since $d(v)$ is an ancestor of $p(v)$ in $T$, $u = d(v)$.  Suppose on the other hand that $u \not= p(v)$.  Both $d(v)$ and $d(p(v))$ are ancestors of $u$ in $T$.  Suppose there is a path from $s$ to $p(v)$ that avoids $d(v)$.  Adding $(p(v), v)$ to this path produces a path from $s$ to $v$ that avoids $d(v)$, a contradiction.  Thus $d(v)$ dominates $p(v)$.  Since $d(v) \not= p(v)$, $d(v)$ dominates $d(p(v))$.  Suppose there is a path from $s$ to $v$ that avoids $d(p(v))$.  Let $(x, v)$ be the last arc on this path.  Then $x$ is a descendant of $u$ and an ancestor of $p(v)$ in $T$.  Replacing $(x, v)$ by the path in $T$ from $x$ to $p(v)$ gives a path from $s$ to $p(v)$ that avoids $d(p(v))$, a contradiction.  Thus $d(p(v))$ dominates $v$.  Since $d(p(v)) \not= v$, $d(p(v))$ dominates $d(v)$.  We conclude that $d(v) = d(p(v))$.
\end{proof}

\begin{lemma}
\label{lemma:contract-1}
Let $v \not= s$ be a vertex with no entering cross arc or back arc, let $(u, v)$ be an arc entering $v$ with $u$ maximum, and suppose $d(x) \ge u$ for every descendant $x$ of $u$ in $T$. Let $G'$ be the graph formed from $G$ by contracting $v$ into $p(v)$.  Then the dominator tree $D'$ of $G'$ is $D$ with $v$ and the arc into $v$ deleted.
\end{lemma}
\begin{proof}
Since $d(x) \ge u \ge p(v)$ for every descendant $x$ of $v$ in $T$, $v$ dominates no vertices in $G$ except itself. Thus $v$ is a leaf in $D$. Let $x$ and $w$ be vertices in $G'$.  Suppose $x$ does not dominate $w$ in $G'$.  Then there is a simple path $P$ from $s$ to $w$ in $G'$ that avoids $x$, which is either a path in $G$ or can be converted into a path in $G$ that avoids $x$ as follows. Replace any new arc from $p(v)$ by a pair of old arcs, one into $v$ and one out of $v$. If $x$ is not a proper descendant of $u$, replace any new arc $(y, p(v))$ by the path in $T$ from $y$ to $p(v)$: $y \not= x$, and $y$ must be a descendant of $u$ and a proper ancestor of $p(v)$.  If $x$ is a proper descendant of $u$, replace the part of $P$ from $s$ to $p(v)$ by a path in $G$ from $s$ to $p(v)$ that avoids $x$. Such a path must exist since $d(p(v)) \ge u$.  Thus $x$ does not dominate $w$ in $G$.

Conversely, suppose $x$ does not dominate $w$ in $G$. Then there is a simple path $P$ from $s$ to $w$ in $G$ that avoids $x$.  If $v$ is not on $P$, $P$ is a path in $G'$.  If $v$ is on $P$, $P$ can be converted to a path from $s$ to $w$ in $G'$ that avoids $x$ as follows.  Let $T'$ be the spanning tree in $G'$ formed from $T$ by the contraction, and let $(y, v)$ and $(v, z)$ be the arcs into and out of $v$ on $P$, respectively.  If $x \not= p(v)$, replace $(y, v)$ and $(v, z)$ on $P$ by $(y, p(v))$ and $(p(v), z)$.  Suppose $x = p(v)$.  Then $u \not= p(v)$, since $u = p(v)$ implies by Lemma \ref{lemma:parent-dom} that $u = x$ dominates $v$ in $G$, but $P$ contains a path from $s$ to $v$ that avoids $x$.  If $z$ is not a descendant of $u$, in $T$, then $z$ is not a descendant of $u$ in $T'$.  Replace the part of $P$ from $s$ to $z$ by the path in $T'$ from $s$ to $z$.  If $z$ is a descendant of $u$ in $T$, replace the part of $P$ from $s$ to $z$ by a path in $G$ from $s$ to $z$ that avoids $x$, which must exist since $d(z) \ge u > p(v) = x$.  Thus $x$ does not dominate $w$ in $G'$.

We conclude that the dominators of any vertex $w \not= v$ are the same in $G$ and $G'$.
\end{proof}

Algorithm AD uses Lemmas \ref{lemma:acyclic-contract}, \ref{lemma:parent-dom}, and \ref{lemma:contract-1} to find dominators.  It chooses contractions by marking arcs.  When the last arc $(u, v)$ into a vertex $v$ is marked, it contracts $v$ into $p(v)$, thereby deleting $v$.  Each new arc created by a contraction is marked if and only if the arc it replaces is marked.  Figure \ref{fig:acyclic-2} illustrates how Algorithm AD works.\\

\begin{figure}[h]
\begin{center}
\fbox{
\begin{minipage}[h]{16cm}
\begin{center}
\textbf{Algorithm AD: Find Dominators in an Acyclic Graph, Version 1}
\end{center}
\begin{description}\setlength{\leftmargin}{10pt} \setlength{\itemsep}{0pt} \parskip0pt \parsep0pt
\item[Initialization:] Find a spanning tree $T$ of the input graph $G$, let $p$ be the parent function of $T$, number the vertices of $T$ from $1$ to $n$ in a bottom-up order, and identify vertices by number. Unmark all arcs of $G$. Assign $\mathit{same}(u) \leftarrow \{u\}$ for each vertex $u$.\\
\item[Main Loop:]
\textbf{for} $u = 1$ \textbf{until} $n$ \textbf{do} \hfill
\item[] \hspace{12.5ex} \textbf{while} some tree or forward arc $(u, v)$ is unmarked \textbf{do} \hfill
\item[] \hspace{12.5ex} $\{$ \ \ mark $(u, v)$; \hfill
\item[] \hspace{16ex} \textbf{if} all arcs into $v$ are marked \textbf{then}\hfill
\item[] \hspace{16ex} $\{$ \ \ \textbf{if} $u=p(v)$ \textbf{then} \textbf{for} $w \in \emph{same}(v)$ \textbf{do} $d(w) \leftarrow u$
\item[] \hspace{19.5ex} \textbf{else} $\emph{same}(p(v)) \leftarrow \emph{same}(p(v)) \cup \emph{same}(v)$;
\ignore{\item[] \hspace{14ex} $\{$ \ \ $\emph{rd}(v) \leftarrow p(v)$; $b(v) \leftarrow$ \textbf{if} $u = p(v)$ \textbf{then} \textbf{true} \textbf{else} \textbf{false};\hfill}
\item[] \hspace{19.5ex} contract $v$ into $p(v)$ \ \ $\}$ \ \ $\}$\\
\ignore{
\item[Pass 3:] \textbf{for} $u = 2$ \textbf{until} $n$ \textbf{do} $d(v) \leftarrow$ \textbf{if} $b(v)$ \textbf{then} $\emph{rd}(v)$ \textbf{else} $d(\emph{rd}(v))$}
\end{description}
\end{minipage}
}
\end{center}
\end{figure}

\begin{figure}
\begin{center}
\scalebox{0.7}[0.7]{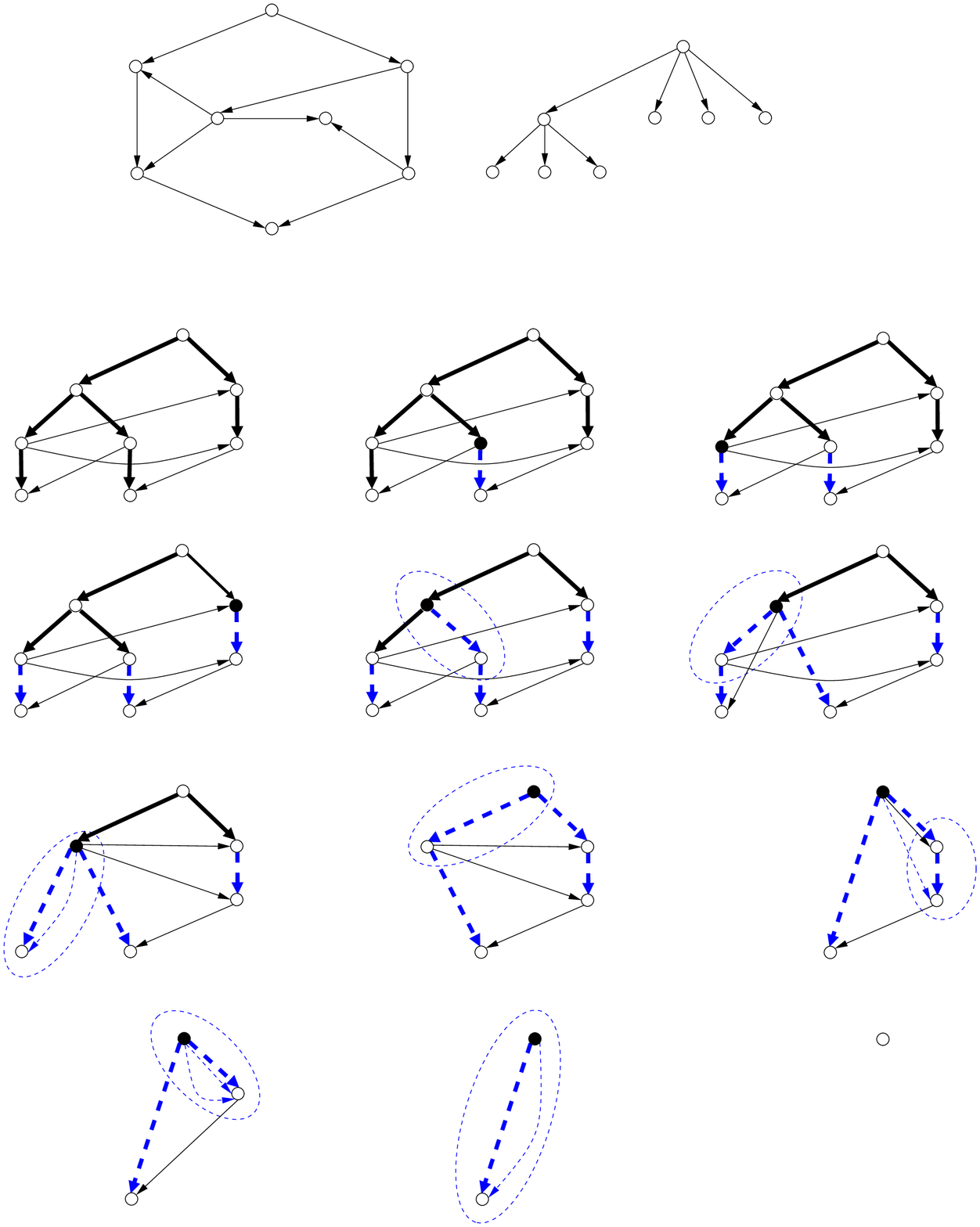}
\end{center}
\caption{\label{fig:acyclic-2} Execution of algorithm AD on an acyclic flow graph. Tree arcs are shown bold; the vertex $u$ processed in the main loop is shown filled; marked arcs are shown dashed; loop arcs created by contractions are omitted.}
\end{figure}

\begin{remark}
During the main loop of Algorithm AD, just before contracting a vertex $v$ into $p(v)$, we can store a \emph{relative dominator}  $\mathit{rd}(v) \leftarrow p(v)$ and an indicator bit $b(v)$ for $v$, where we set $b(v) \leftarrow \textbf{true}$ if $u = p(v)$, and $b(v) \leftarrow \textbf{false}$ otherwise.
This eliminates the need to maintain the $\mathit{same}$-sets but results in a three pass algorithm. The third pass processes vertices in a top-down order (the reverse of a bottom-up order) and for each vertex $v$ it assigns $d(v) \leftarrow \mathit{rd}(v)$ if $b(v) = \textbf{true}$, and $d(v) \leftarrow d(\mathit{rd}(v))$ otherwise. The LT algorithm computes immediate dominators in the same way, but it uses a different, more static definition of relative dominators.\\
\end{remark}

To prove that Algorithm AD is correct, our main task is to verify that the hypotheses of Lemmas \ref{lemma:parent-dom} and \ref{lemma:contract-1} hold for each contraction, and that the main loop deletes all vertices except $s$. The bottom-up order of vertex processing guarantees these properties.  The correctness of the dominators computation follows immediately from Lemma \ref{lemma:parent-dom}.

\begin{lemma}
\label{lemma:acyclic-loop}
If the original graph is acyclic, then throughout the main loop the current graph is acyclic, and every loop arc is marked.
\end{lemma}
\begin{proof}
Contractions preserve acyclicity by Lemma \ref{lemma:acyclic-contract}.  By assumption the original graph contains no loop arcs.  By Lemma \ref{lemma:acyclic-contract}, any new loop arc $(p(v), p(v))$ is created by a contraction of $v$ into $p(v)$ and replaces a former arc $(p(v), v)$. For such a contraction to occur, $(p(v),v)$ must be marked. The new loop arc inherits the mark.
\end{proof}

\begin{lemma}
\label{lemma:marked-arc}
Throughout the main loop, each marked arc $(x, y)$ is a tree, forward, or loop arc such that $x \le u$.
\end{lemma}
\begin{proof}
Each arc $(u, v)$ marked in the main loop satisfies the lemma when it is marked, and it continues to do so since vertices are processed in increasing order.  All new arcs into $p(v)$ created by contracting $v$ into $p(v)$ are tree, forward, or loop arcs that satisfy the lemma when added and continue to do so.  If $(v, y)$ is a tree, forward, or loop arc before $v$ is contracted into $p(v)$, then so is new arc $(p(v), y)$, and $p(v) \le u$, so if $(v, y)$ is marked its replacement $(p(v), y)$ satisfies the lemma and continues to do so.
\end{proof}

The next lemma is the main part of the correctness proof.

\begin{lemma}
\label{lemma:invariants}
For any $u$, during iteration $u$ of the main loop, every undeleted vertex $v \not= s$ has $d(v) \ge u$.  At the end of the iteration, every undeleted vertex $v \not= s$ has $d(v) > u$.  Each $v$ contracted into $p(v)$ during iteration $u$ is such that $u$ and $v$ satisfy the hypotheses of Lemmas \ref{lemma:parent-dom}  and \ref{lemma:contract-1}.
\end{lemma}
\begin{proof}
Since $v \ge 1$ for all $v$, the first invariant holds at the beginning of the first iteration.  If the second invariant holds at the end of an iteration, then the first invariant holds at the beginning of the next. Suppose the first invariant holds before a contraction of $v$ into $p(v)$ in iteration $u$.  By the first invariant, $d(x) \ge u$ for every descendant $x$ of $u$ in $T$.  All arcs into $v$ are marked, so by Lemma \ref{lemma:marked-arc} they are all tree, forward, or loop arcs from vertices no less than $u$.  The last one marked is $(u, v)$, so all the hypotheses of Lemmas \ref{lemma:parent-dom} and \ref{lemma:contract-1} hold.  By Lemma \ref{lemma:contract-1}, the contraction preserves the dominators of the undeleted vertices, so the first invariant holds after the contraction.  By induction on the number of contractions, the first and third invariants hold during the entire iteration.

Consider a time during the iteration when there is a vertex $v$ such that $d(v) = u$.  Among such vertices, choose one, say $v$, such that there is no path from another such vertex to $v$.  Such a $v$ must exist since the current graph is acyclic by Lemma \ref{lemma:acyclic-loop}.  Suppose $v$ has an entering non-loop arc $(x, v)$ with $x \not= u$.  By the choice of $v$, $x = s$ or $d(x) > u$.  In either case there is a path from $s$ to $v$ that avoids $u$, contradicting $d(v) = u$.  We conclude that $u = p(v)$, every arc entering $v$ is a tree arc or loop arc, and $v$ will be deleted during iteration $u$ once all incoming tree arcs are marked.  Hence there can be no such undeleted $v$ at the end of the iteration, making the second invariant true.
\end{proof}

\begin{corollary}
\label{corollary:start-vertex}
At the end of the main loop, only $s$ is undeleted.
\end{corollary}
\begin{proof}
Once the second invariant of Lemma \ref{lemma:invariants} holds for $u = n$, $s = n$ is the only undeleted vertex, since $d(x) \le n$ if $x \not= s$.
\end{proof}

\begin{theorem}
\label{theorem:AD-correct}
Algorithm AD is correct.
\end{theorem}
\begin{proof}
The theorem is immediate from Lemmas \ref{lemma:parent-dom}, \ref{lemma:contract-1}, \ref{lemma:invariants}, and Corollary \ref{corollary:start-vertex}.
\end{proof}

Now we develop an efficient implementation of Algorithm AD.  The first step is to observe that keeping track of unmarked arcs suffices.  Instead of marking arcs, we delete them.  When a vertex $v$ loses its last incoming arc, we contract it into $p(v)$.  This contraction adds no new arcs into $p(v)$, and in particular creates no loops.

\begin{theorem}
\label{theorem:AD-correct-2}
Algorithm AD remains correct if each contraction of a vertex $v$ into its parent is preceded by deleting all arcs into $v$.
\end{theorem}
\begin{proof}
Consider a run of the modified version of Algorithm AD.  Rerun the original version of the algorithm, choosing the same spanning tree and the same vertex numbering, and making the same choices of arcs to mark as in the run of the modified algorithm.  An induction on the number of arcs marked shows that this is always possible.  The two runs compute the same immediate dominators.
\end{proof}

In algorithm AD with arcs deleted as described, the number of arcs never increases, and only the first end of an arc changes.  Each new first end is a proper ancestor in the original spanning tree of the old first end.  To keep track of the effect of contractions, we use a \emph{disjoint set data structure} \cite{dsu:tarjan}.  Such a data structure maintains a collection of disjoint sets, each with a name, under three operations:

\begin{description}
\item[\emph{make-set}$(x)$:] Create a new set $\{x\}$ with name $x$.  Element $x$ must be in no existing set.
\item[\emph{find}$(x)$:] Return the name of the set containing element $x$.
\item[\emph{unite}$(x,y)$:] Unite the sets containing elements $x$ and $y$ and give the new set the name of the old set containing $x$.
\end{description}

We maintain the original arc ends and the parent function $p$ of the original spanning tree, and use the disjoint set data structure to map these to the current graph and the current tree.  To initialize the sets, we perform \emph{make-set}$(v)$ for every vertex $v$. When a vertex $v$ is contracted into its current parent, we perform \emph{unite}$(p(v), v)$.   If $(x, y)$ is an original arc, the corresponding current arc is $(\emph{find}(x), y)$.  If $p$ is the original parent function, the current parent of undeleted vertex $v \not= s$ is \emph{find}$(p(v))$.

To determine when to do contractions, we maintain an integer \emph{total}$(v)$ for each vertex $v \not= s$, equal to the number of undeleted arcs into $v$. When \emph{total}$(v)$ reaches zero, we contract $v$ into its current parent.

The last thing we need is a way to keep track of tree and forward arcs.  In the main loop we add an arc to the graph only once it is guaranteed to become a tree or forward arc by the time it is a candidate for marking.  For any two vertices $x$ and $y$, let \emph{nca}$(x, y)$ be their nearest common ancestor in the original spanning tree $T$.  An original cross arc $(x, y)$ cannot become a tree or forward arc until $u$ in the current iteration of the main loop is at least \emph{nca}$(x, y)$.  An original tree or forward arc $(x, y)$ has \emph{nca}$(x, y) = x$.  This means that the algorithm can only mark the current arc corresponding to $(x, y)$ once $u \ge \mathit{nca}(x, y)$.  We add $(x, y)$ to the graph at the beginning of iteration $u$.  Specifically, we maintain, for each undeleted vertex $x$, a bag (multiset) \emph{out}$(x)$ of vertices $v$ such that $(x, v)$ is a current undeleted arc with \emph{nca}$(x, v) \le u$, where $u$ is the current iteration of the main loop.  This bag is initially empty.  (We denote an empty bag by ``$\left [ \  \right ]$''.)  At the beginning of iteration $u$ of the main loop, for each original arc $(x, v)$ such that \emph{nca}$(x, v) = u$, we add $v$ to $\mathit{out}(\mathit{find}(x))$.  When contracting a vertex $v$ into its current parent $x$, we replace $\mathit{out}(x)$ by $\mathit{out}(x) \cup \mathit{out}(v)$.  This guarantees that during iteration $u$ of the main loop, $\mathit{out}(u)$ will contain a copy of vertex $v$ for each undeleted tree or forward arc $(u, v)$, and no other vertices.

Combining these ideas produces Version 2 of Algorithm AD.

\begin{figure}[h]
\begin{center}
\fbox{
\begin{minipage}[h]{16cm}
\begin{center}
\textbf{Algorithm AD: Find Dominators in an Acyclic Graph, Version 2}
\end{center}
\begin{description}\setlength{\leftmargin}{10pt} \setlength{\itemsep}{0pt} \parskip0pt \parsep0pt
\item[Initialization:] Find a spanning tree $T$ of the input graph $G$, let $p$ be the parent function of $T$, number the vertices of $T$ in a bottom-up order, and identify vertices by number.  Let $A$ be the arc set of $G$.  Compute $\emph{total}(u) \leftarrow | \{ (x, u) \in A \} |$ and $\emph{arcs}(u) \leftarrow \{ (x, y) \in A \ | \ \emph{nca}(x, y) = u \}$ for each vertex $u$.\\
\item[Main Loop:]
\textbf{for} $u = 1$ \textbf{until} $n$ \textbf{do} \hfill
\item[] \hspace{12.5ex}  $\{$ \ \ $\emph{out}(u) \leftarrow [ \ ]$; \hfill
\item[] \hspace{16ex} $\mathit{make\text{-}set}(u)$; \hfill
\item[] \hspace{16ex} $\mathit{same}(u) \leftarrow \{ u \}$; \hfill
\item[] \hspace{16ex} \textbf{for} $(x, y) \in \mathit{arcs}(u)$ \textbf{do} add $y$ to $\emph{out}(\emph{find}(x))$; \hfill
\item[] \hspace{16ex} \textbf{while} $\emph{out}(u) \not= [ \ ]$ \textbf{do} \hfill
\item[] \hspace{16ex} $\{$ \ \ delete some $v$ from \emph{out}$(u)$; $\emph{total}(v) \leftarrow \emph{total}(v) - 1$; \hfill
\item[] \hspace{19.5ex} \textbf{if} $\emph{total}(v) = 0$ \textbf{then} \hfill
\item[] \hspace{19.5ex} $\{$ \ \ $x \leftarrow \emph{find}(p(v))$;
\item[] \hspace{23ex} \textbf{if} $u = x$ \textbf{then} \textbf{for} $w \in \emph{same}(v)$ \textbf{do} $d(w) \leftarrow u$
\item[] \hspace{23ex} \textbf{else} $\emph{same}(x) \leftarrow \emph{same}(x) \cup \emph{same}(v)$; \hfill
\item[] \hspace{23ex} $\emph{unite}(p(v), v)$;
\item[] \hspace{23ex} $\emph{out}(x) \leftarrow \emph{out}(x) \cup \emph{out}(v)$ \ \ $\}$ \ \ $\}$ \ \ $\}$\\
\end{description}
\end{minipage}
}
\end{center}
\end{figure}

In the main loop, the \emph{unite} operation and the union of bags implement the contraction of $v$ into its current parent $x$.  The first argument of the \emph{unite} could be $x$, but by making it $p(v)$ we make the sequence of set operations an instance of static tree set union, in which the set of \emph{unite} operations is known in advance but their order is determined on-line.  The main loop does exactly one operation $\emph{unite}(p(v), v)$ for each vertex $v \not= s$, but their order depends on the non-tree arcs.

\begin{lemma}
\label{lemma:acyclic-pass2}
Suppose $G$ is acyclic.  During the \textbf{while} loop in iteration $u$ of the main loop, for each original arc $(x, v)$ such that $\mathit{nca}(x, v) \le u$ and whose corresponding current arc $(\mathit{find}(x), v)$ is undeleted, $\mathit{out}(\mathit{find}(x))$ contains one copy of $v$, and such vertices are the only vertices in $\mathit{out}$-bags.
\end{lemma}
\begin{proof}
The proof is by induction on the number of steps in the main loop.  The inner \textbf{for} loop in iteration $1$ establishes the invariant for iteration $1$.  If the invariant holds at the end of iteration $u$, then the inner \textbf{for} loop in iteration $u+1$ establishes it for iteration $u+1$.  Deletions of vertices from \emph{out}-bags correspond to deletions of the corresponding arcs, so such deletions preserve the invariant.  The \emph{unite} operations and bag unions done to implement contractions also preserve the invariant.
\end{proof}

\begin{corollary}
\label{corollary:acyclic-pass2}
During the \textbf{while} loop in iteration $u$ of the main loop, for each undeleted tree or forward arc $(u, v)$, there is a copy of $v$ in $\mathit{out}(u)$, and such vertices are the only vertices in $\mathit{out}(u)$.
\end{corollary}
\begin{proof}
If $(u, v)$ is a current undeleted tree or forward arc, it corresponds to an original arc $(x, v)$ such that $x$ and $v$, and hence $\emph{nca}(x, v)$, are descendants of $u$ in the original spanning tree $T$.  Since the vertex order is bottom-up, $\emph{nca}(x, v) \le u$.  By Lemma \ref{lemma:acyclic-pass2}, there is a copy of $v$ in $\emph{out}(u)$ corresponding to $(x, v)$.  Conversely, if $v$ is in $\emph{out}(u)$ during the \textbf{while} loop in iteration $u$, then by Lemma \ref{lemma:acyclic-pass2} there is an undeleted current arc $(u, v)$ corresponding to an original arc $(x, v)$ such that $\emph{nca}(x, v) \le u$.  Since the vertex numbering is bottom-up, $v$ must be a descendant of $u$.  Thus $(u, v)$ is a tree or forward arc.
\end{proof}

\begin{theorem}
\label{theorem:AD2-correct}
Version 2 of Algorithm AD is correct.
\end{theorem}
\begin{proof}
The Theorem is immediate from Theorem \ref{theorem:AD-correct-2} and Corollary \ref{corollary:acyclic-pass2}.
\end{proof}

To implement version 2 of Algorithm AD, we represent each set $\mathit{same}(v)$ and each bag $\mathit{out}(v)$ by a singly-linked circular list.  The circular linking allows unions to be done in $O(1)$ time. Since each vertex is in only one $\mathit{same}$-set, the lists representing these sets can be endogenous. The lists representing $\mathit{out}$-bags must be exogenous, since a vertex can be in several bags, or even in the same bag several times.  Alternatively, the $\mathit{out}$-bags can be represented by endogenous lists of the corresponding arc sets. For a discussion of endogenous and exogenous lists, see \cite{algorithms:tarjan}.
Each vertex and arc is examined $O(1)$ times.  Not counting the nearest common ancestor computations in the initialization, the running time of the algorithm is $O(m)$ plus the time for $n - 1$ \emph{unite} operations and at most $m + n$ \emph{find} operations.  If \emph{unite} and \emph{find} are implemented using compressed trees with appropriate heuristics \cite{dsu:tarjan}, the total time for the \emph{unite} and \emph{find} operations is $O(m \alpha(n, m/n))$.  Furthermore the set of \emph{unite} operations is known in advance, although their sequence is not.  This makes the set operations an instance of the static tree disjoint set union problem, which is solvable in $O(m)$ time on a RAM \cite{dsu:gt}.  The computation of nearest common ancestors can also be done by solving an instance of the static tree disjoint set union problem \cite{lca:ahu,dsu:gt}.  We conclude that the overall running time of Version 2 of Algorithm AD is $O(m \alpha(n, m/n))$, or $O(m)$ on a RAM, depending on the implementation. The only data structure needed other than simple lists and maps is one to maintain disjoint sets.

Our final refinement of Algorithm AD, Version 3, eliminates the need to compute nearest common ancestors.  We accomplish this by choosing the spanning tree and vertex order carefully.  Specifically, we choose a \emph{depth-first spanning tree} and a corresponding \emph{reverse preorder} \cite{dfs:t}.  Such a tree and order have the property that every tree or forward arc $(v, w)$ is has $v > w$ and every cross or back arc $(v, w)$ has $v < w$ \cite{dfs:t}.  We insert vertices into $\mathit{out}$ bags as follows.  At the beginning of iteration $u$ of the main loop, for every arc $(x, u)$, we insert $u$ into $\mathit{out}(\mathit{find}(x))$.  If $(x, u)$ is a tree or forward arc, $x$ has not yet been processed in main loop, so $\mathit{find}(x) = x$, and when $x$ is processed later, $u$ will be in $\mathit{out}(x)$ as desired.  If $(x, u)$ is a cross arc, $x$ has already been processed in the main loop.  Vertex $u$ will remain in $\mathit{out}(\mathit{find}(x))$ (which changes as $\mathit{find}(x)$ changes) until $\mathit{nca}(x, u)$ is processed, at which time Version 2 of Algorithm AD would add $u$ to $\mathit{out}(\mathit{find}(x))$.  Thus, even though the new version adds vertices to \emph{out} bags sooner than Version 2, these early additions do not change the candidates for arc deletions, making Version 3 correct.

\begin{figure}[h]
\begin{center}
\fbox{
\begin{minipage}[h]{16cm}
\begin{center}
\textbf{Algorithm AD: Find Dominators in an Acyclic Graph, Version 3}
\end{center}
\begin{description}\setlength{\leftmargin}{10pt} \setlength{\itemsep}{0pt} \parskip0pt \parsep0pt
\item[Initialization:] Do a depth-first search of $G$ to generate a depth-first spanning tree $T$ of $G$ with parent function $p$ and to number the vertices in reverse preorder with respect to the search.  Identify vertices by number.  Let $A$ be the arc set of $G$.\\
\item[Main Loop:]
\textbf{for} $u = 1$ \textbf{until} $n$ \textbf{do} \hfill
\item[] \hspace{12.5ex}  $\{$ \ \ $\mathit{total}(u) \leftarrow 0$; \hfill
\item[] \hspace{16ex} $\mathit{out}(u) \leftarrow [ \ ]$; \hfill
\item[] \hspace{16ex} $\mathit{make\text{-}set}(u)$; \hfill
\item[] \hspace{16ex} $\mathit{same}(u) \leftarrow \{v\}$; \hfill
\item[] \hspace{16ex} \textbf{for} $(x, u) \in A$ \textbf{do}\hfill
\item[] \hspace{16ex} $\{$ \ \ $\mathit{total}(u) \leftarrow \mathit{total}(u) + 1$; add $u$ to $\mathit{out}(\mathit{find}(x))$ \ \ $\}$ \hfill
\item[] \hspace{16ex} \textbf{while} $\mathit{out}(u) \not= [ \ ]$ \textbf{do} \hfill
\item[] \hspace{16ex} $\{$ \ \ delete some $v$ from $\mathit{out}(u)$; $\mathit{total}(v) \leftarrow \mathit{total}(v) - 1$; \hfill
\item[] \hspace{19.5ex} \textbf{if} $\mathit{total}(v) = 0$ \textbf{then} \hfill
\item[] \hspace{19.5ex} $\{$ \ \ $x \leftarrow \mathit{find}(p(v))$;
\item[] \hspace{23ex} \textbf{if} $u=x$ \textbf{then} \textbf{for} $w \in \mathit{same}(v)$ \textbf{do} $d(w) \leftarrow u$ \hfill
\item[] \hspace{23ex} \textbf{else} $\mathit{same}(x) \leftarrow \mathit{same}(x) \cup \mathit{same}(v)$;
\item[] \hspace{23ex} $\mathit{unite}(p(v), v)$;
\item[] \hspace{23ex} $\mathit{out}(x) \leftarrow \mathit{out}(x) \cup \mathit{out}(v)$ \ \ $\}$ \ \ $\}$ \ \ $\}$\\
\end{description}
\end{minipage}
}
\end{center}
\end{figure}

\begin{lemma}
\label{lemma:acyclic3-pass2}
Suppose $G$ is acyclic.  During the \textbf{while} loop in iteration $u$ of the main loop, for each original arc $(x, v)$ such that $v \le u$ and whose corresponding current arc $(\mathit{find}(x), v)$ is undeleted, $\mathit{out}(\mathit{find}(x))$ contains one copy of $v$, and such vertices are the only vertices in out bags.
\end{lemma}
\begin{proof}
The proof is analogous to the proof of Lemma \ref{lemma:acyclic-pass2}.
\end{proof}

\begin{corollary}
\label{lemma:corollary3-pass2}
During the \textbf{while} loop in iteration $u$ of the main loop, for each undeleted tree or forward arc $(u, v)$, there is a copy of $v$ in $\mathit{out}(u)$, and such vertices are the only vertices in $\mathit{out}(u)$.
\end{corollary}
\begin{proof}
If $(u, v)$ is a current undeleted tree or forward arc, it replaces an original arc $(x, v)$ such that $x$ and $v$ are descendants of $u$ in the original spanning tree $T$.  Since the vertex order is bottom-up, $v \le u$.  By Lemma \ref{lemma:acyclic3-pass2}, there is a copy of $v$ in $\emph{out}(u)$ corresponding to $(x, v)$.  Conversely, if $v$ is in $\emph{out}(u)$ during the \textbf{while} loop in iteration $u$, then by Lemma \ref{lemma:acyclic3-pass2} there is an undeleted current arc $(u, v)$ replacing an original arc $(x, v)$ such that $v \le u$.  Since $(u, v)$ is the current arc replacing $(x, v)$, $u$ is an ancestor of $x$ in $T$. Since the vertex numbering is reverse preorder, every ancestor of $x$ in $T$ that is not an ancestor of $v$ in $T$ has number less than $v$.  Since $u$ is an ancestor of $x$ and $u \ge v$, $u$ is an ancestor of $v$.  Thus $(u, v)$ is a tree or forward arc.
\end{proof}

\begin{theorem}
\label{theorem:AD3-correct}
Version 3 of Algorithm AD is correct.
\end{theorem}
\begin{proof}
The Theorem is immediate from Theorem \ref{theorem:AD-correct-2} and Corollary \ref{lemma:corollary3-pass2}.
\end{proof}

If we use a different vertex order in the main loop, namely postorder, then we can fold the main loop into the depth-first search that builds the spanning tree. The result is a one-pass algorithm to find dominators. Unfortunately this method must compute nca's to determine when to add vertices to $\mathit{out}$-bags, so it uses two disjoint set data structures concurrently, one to keep track of contractions and the other to compute nca's.  We discuss this approach more fully in Section \ref{sec:general}, since it can be used for general graphs as well.

\section{Finding Dominators in a General Graph}
\label{sec:general}

As discussed in Section \ref{sec:introduction}, Ramalingam~\cite{loops:Ramalingam} gave a reduction of the dominator-finding problem on a general graph to the same problem on an acyclic graph.  His reduction uses simple data structures and static-tree disjoint set union, so it has the same asymptotic time bound as Algorithm AD.  We give a streamlined version of his reduction in Section \ref{sec:loops}.  By combining this reduction or his original reduction with Algorithm AD, we obtain an algorithm that finds dominators in an arbitrary graph in near-linear or linear time and uses only simple data structures and static-tree disjoint set union.  Although this achieves our goal, we prefer an algorithm that is self-contained and as simple as possible.  We develop such an algorithm in this section.

To explain the algorithm, we need some terminology about strongly connected subgraphs.  Let $T$ be an arbitrary spanning tree of $G$ rooted at $s$, let $p$ be the parent function of $T$, and suppose the vertices of $T$ are numbered from $1$ to $n$ in a bottom-up order and identified by number.  If $u$ is any vertex, the \emph{loop} of $u$, denoted by $\emph{loop}(u)$, is the set of all descendants $x$ of $u$ in $T$ such that there is a path from $x$ to $u$ containing only descendants of $u$ in $T$.  Vertex $u$ is the \emph{head} of the loop.  The loop of $u$ induces a strongly connected subgraph of $G$ (every vertex is reachable from any other), and it is the unique maximal set of descendants of $u$ that does so.  If $u$ and $v$ are any two vertices, their loops are either disjoint or nested (one is contained in the other).  The \emph{loop nesting forest} $H$ is the forest with parent function $h$ such that $h(v)$ is the nearest proper ancestor $u$ of $v$ in $T$ whose loop contains $v$ if there is such a vertex, \textbf{null} otherwise.  If $u$ is any vertex, $\emph{loop}(u)$ is the set of all descendants of $u$ in $H$.
An \emph{entry} to $\mathit{loop}(u)$ is an arc $(v, w)$ such that $w$ is in $\mathit{loop}(u)$ but $v$ is not; $(v, w)$ is a head entry if $w = u$ and a non-head entry otherwise.
An \emph{exit} from $\emph{loop}(u)$ is an arc from a vertex in $\emph{loop}(u)$ to a vertex in $\emph{loop}(h(u)) - \emph{loop}(u)$.  A loop has an exit if and only if it is contained in a larger loop.  These definitions extend to an arbitrary spanning tree the corresponding definitions for a depth-first spanning tree \cite{loops:Ramalingam,st:t}. See Figure \ref{fig:dfs-loops}.

\begin{figure}[t]
\begin{center}
\scalebox{0.7}[0.7]{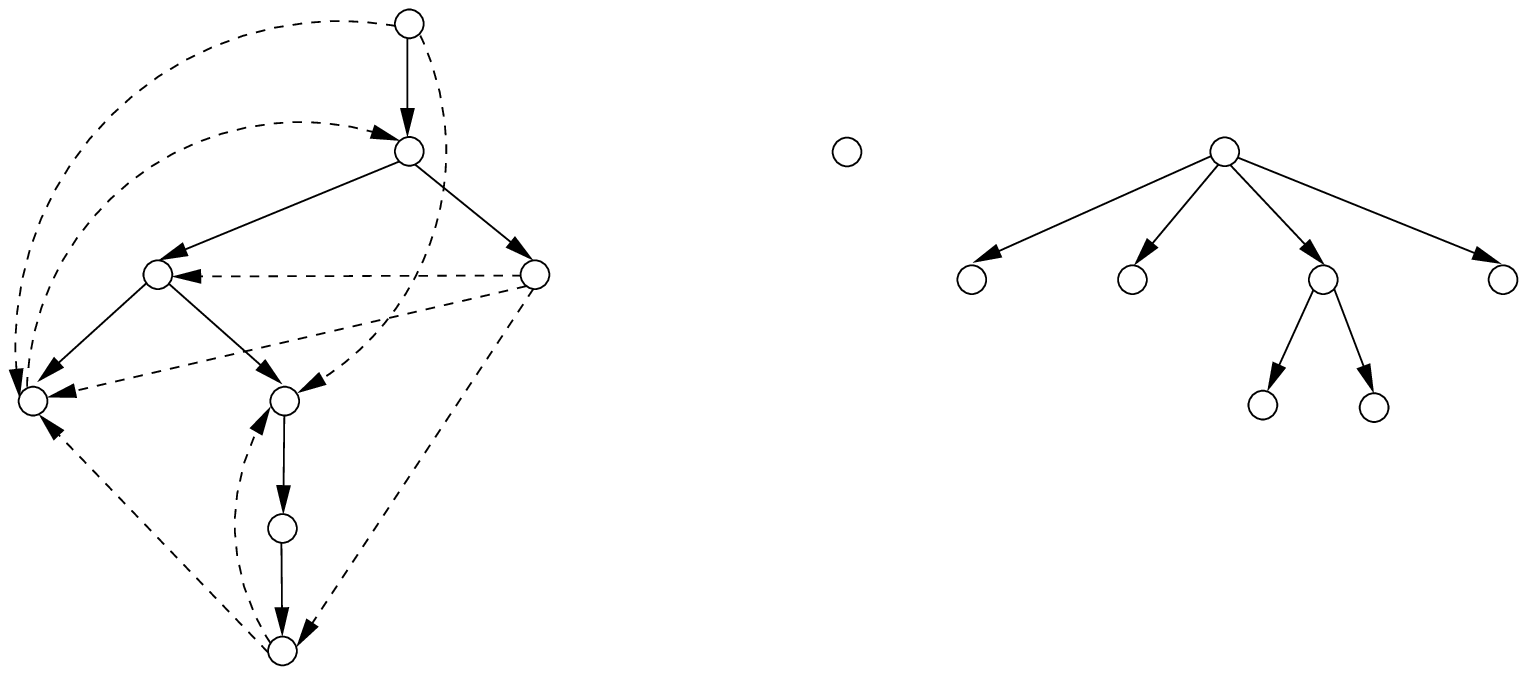}
\end{center}
\vspace{-0.5cm}
\caption{\label{fig:dfs-loops} A depth-first spanning tree of the flow graph in Figure \ref{fig:dominator-tree} (tree are shown solid, non-tree arcs are shown dashed), and the corresponding loop nesting forest.}
\end{figure}

A loop is \emph{reducible} if all its entries enter its head; that is, it has no non-head entries.
The head of a reducible loop dominates all vertices in the loop.  A flow graph is \emph{reducible} \cite{rdflow:hu74,reducibility:jcss:tarjan} if all its loops are reducible.  If $G$ is reducible, deletion of all its back arcs with respect to any spanning tree produces an acyclic graph with the same dominators as $G$.  Thus Algorithm AD extends to find the dominators of any reducible graph.

To extend algorithm AD to general graphs, we need a way to delete vertices on cycles.  For this purpose we use the \emph{transform} operation, which adds certain arcs to the graph and then does a contraction. The operation $\emph{transform}(u, v)$ requires that $u$ be a proper ancestor of $v$ in $T$ and consists of the following two steps:

\begin{description}
\item[Step 0:] For each arc from $v$ to a proper descendant $w \not\in \{v, p(v)\}$ of $u$ in $T$, add an arc from $p(u)$ to $w$.
\item[Step 1:] Contract $v$ into $p(v)$.
\end{description}

\begin{figure}[t]
\begin{center}
\scalebox{0.7}[0.7]{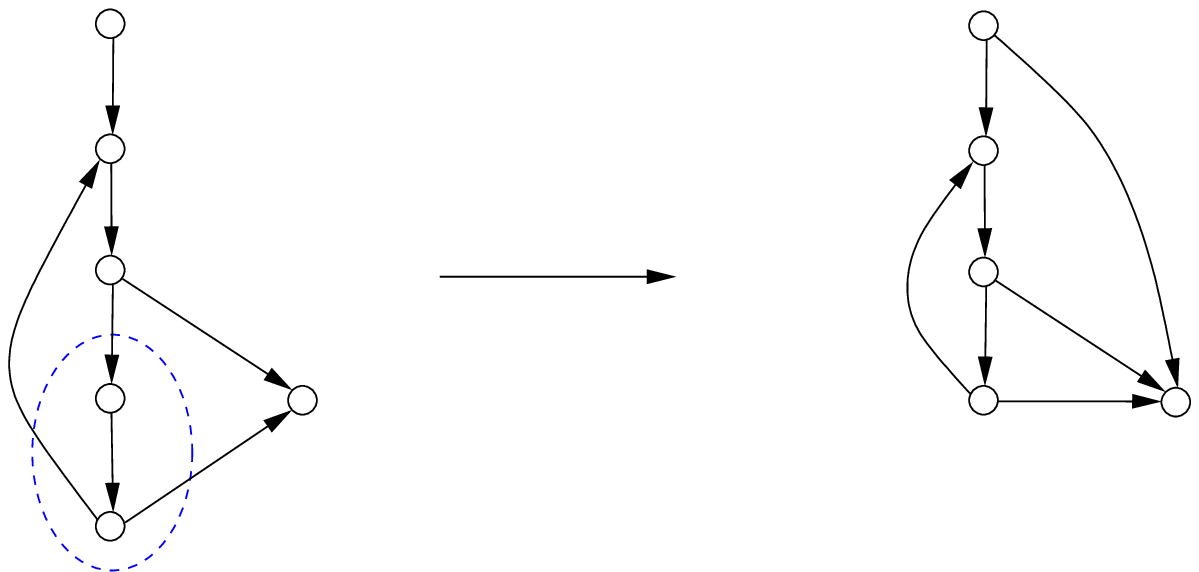}
\end{center}
\vspace{-0.5cm}
\caption{\label{fig:transform} Transform operation.}
\end{figure}

See Figure \ref{fig:transform}.
The effect of the transform is to replace $G$ by a new graph $G'$ and $T$ by a new spanning $T'$.

The purpose of Step 0 is to preserve dominators.  The arcs added in Step 0 play a crucial role in proving the correctness of our algorithm, but the algorithm itself does not actually keep track of such arcs.  This makes the behavior of the algorithm a bit subtle.

The next two lemmas, analogous to Lemmas \ref{lemma:parent-dom} and \ref{lemma:contract-1}, justify the use of transforms.

\begin{lemma}
\label{lemma:loop}
Let $u$ and $v$ be distinct vertices such that $v \in \mathit{loop}(u)$ and $d(v) > u$.  Then $d(v) = d(u)$.
\end{lemma}
\begin{proof}
Neither $u$ nor $v$ dominates the other, since both have immediate dominator greater than $u$.  Let $x \not= u$ be a vertex that does not dominate $v$.  Then there is a path $P$ from $s$ to $v$ that avoids $x$.  If $x$ is a proper descendant of $u$ then $x$ does not dominate $u$.  Suppose $x$ is not a descendant of $u$.  Since $v \in \emph{loop}(u)$, there is a path from $v$ to $u$ containing only descendants of $u$ in $T$, and hence not containing $x$.  Adding this path to $P$ produces a path from $s$ to $u$ that avoids $x$.  Thus $x$ does not dominate $u$.  Conversely, suppose $x$ does not dominate $u$.  Then there is a path $P$ from $s$ to $u$ that avoids $x$.  If $x$ is not a descendant of $u$, adding  to $P$ the path in $T$ from $u$ to $v$ produces a path from $s$ to $v$ that avoids $x$.  If $x$ is a descendant of $u$, there is a path from $s$ to $v$ that avoids $x$ since $d(v) > u$.  Thus $x$ does not dominate $v$.  It follows that $u$ and $v$ have the same proper dominators (dominators other than themselves).
\end{proof}

\begin{lemma}
\label{lemma:loop-transform}
Let $u$ and $v$ be distinct vertices such that $v \in \mathit{loop}(u)$ and $d(x) > u$ for all descendants $x$ of $u$ in $T$. Let $G'$ and $T'$ be the graph and spanning tree formed from $G$ and $T$, respectively, by doing $\mathit{transform}(u, v)$.  Then the dominator tree $D'$ of $G'$ is $D$ with $v$ and its incoming arc $(d(v), v)$ deleted.
\end{lemma}
\begin{proof}
Since $d(x) > u \ge p(v)$ for every descendant $x$ of $v$ in $T$, $v$ dominates no vertices in $G$ except itself.  Thus $v$ is a leaf in $D$.  Let $x$ and $w$ be vertices in $G'$.  Suppose $x$ does not dominate $w$ in $G'$.  Then there is a path $P$ from $s$ to $w$ in $G'$ that avoids $x$.  Suppose $P$ is not in $G$.  Replace a new arc $(p(v), z)$ by old arcs $(p(v), v)$ and $(v, z)$.  If $x$ is not a descendant of $u$, replace a new arc $(p(u), z)$ by the path in $T$ from $p(u)$ to $v$ followed by old arc $(v, z)$, and replace a new arc $(y, p(v))$ by old arc $(y, v)$, followed by a path of descendants of $u$ from $v$ to $u$ (which must exist since $v \in \mathit{loop}(u)$), followed by the path in  $T$ from $u$ to $p(v)$.  The result is a path in $G$ from $s$ to $w$ that avoids $x$.  Suppose $x$ is a descendant of $u$.  If $P$ contains an arc $(p(u), z)$, replace the part of $P$ from $s$ to $z$ by a path in $G$ from $s$ to $v$ that avoids $x$, which must exist since $d(v) > u$, followed by old arc $(v, z)$.  If $P$ contains a new arc $(y, p(v))$, replace the part of $P$ from $s$ to $p(v)$ by a path in $G$ from $s$ to $p(v)$ that avoids $x$, which must exist since $d(p(v)) > u$.  The result is a path in $G$ from $s$ to $w$ that avoids $x$.  Thus $x$ does not dominate $w$ in $G$.

Conversely, suppose $x$ does not dominate $w$ in $G$. Then there is a simple path $P$ from $s$ to $w$ in $G$ that avoids $x$.  If $v$ is not on $P$, $P$ is a path in $G'$.  Suppose $v$ is on $P$.  Let $(y, v)$ and $(v, z)$ be the arcs into and out of $v$ on $P$, respectively.  If $x \not= p(v)$, replace $(y, v)$ and $(v, z)$ on $P$ by $(y, p(v))$ and  $(p(v), z)$.  The result is a path in $G'$ from $s$ to $w$ that avoids $x$. Suppose $x = p(v)$.  If $z$ is not a descendant of $u$, replace the part of $P$ from $s$ to $z$ by the path in $T'$ from $s$ to $z$, which avoids $x$.  If $z$ is a descendant of $u$, replace the part of $P$ from $s$ to $z$ by the path in $T'$ from $s$ to $p(u)$ followed by new arc $(p(u), z)$.  The result is a path in $G'$ from $s$ to $w$ that avoids $x$.  Thus $x$ does not dominate $w$ in $G'$.

We conclude that the dominators of any vertex $w \not= v$ are the same in $G$ and $G'$.
\end{proof}

\begin{figure}[t]
\begin{center}
\scalebox{0.7}[0.7]{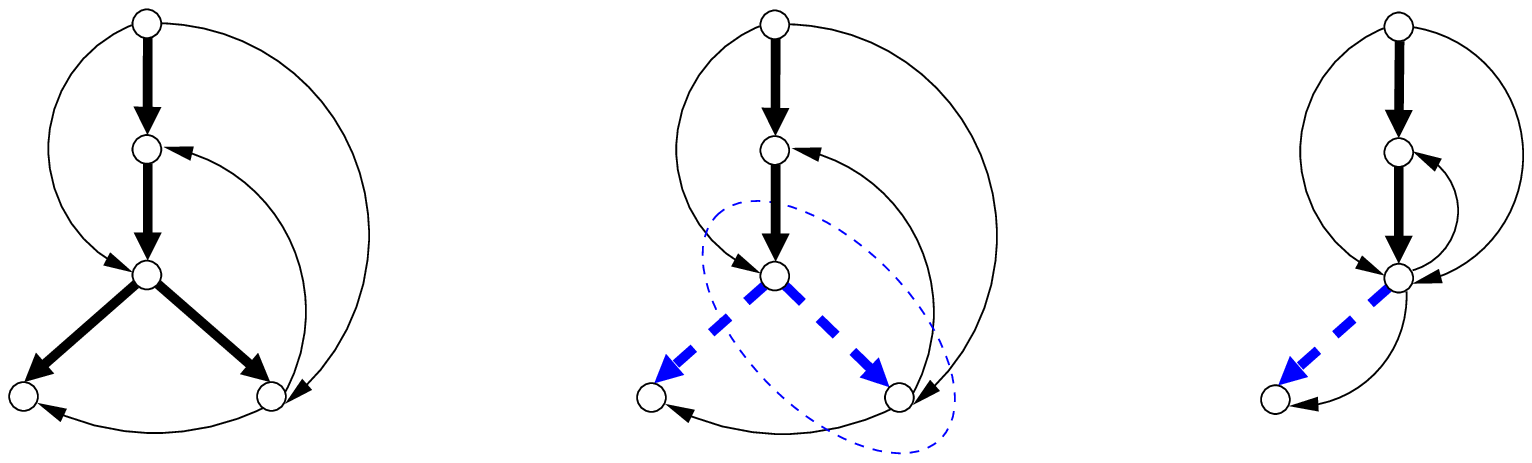}
\end{center}
\vspace{-0.5cm}
\caption{\label{fig:transform-counterexample} Counterexample to Lemma \ref{lemma:loop-transform} if Step 0 is dropped from $\mathit{transform}(u, v)$.  Vertex $5$ is the immediate dominator of all other vertices. Let $u = 4$ and $v = 2$. Contracting $2$ into $3$ without adding an arc into $1$ results in $1$ having immediate dominator $3$ instead of $1$.}
\end{figure}

Lemma \ref{lemma:loop-transform} is false if Step 0 is dropped from $\mathit{transform}(u, v)$, as the example in Figure \ref{fig:transform-counterexample} shows.

For transforms to suffice for deleting vertices on cycles, every cycle must be in a loop. Cycles outside of loops, and indeed without back arcs, can exist if the spanning tree $T$ is arbitrary, but not if $T$ is chosen carefully, specifically if $T$ is a depth-first spanning tree. If so, every cycle contains a back arc \cite{dfs:t}, and more generally every cycle contains a vertex $u$ that is a common ancestor of all other vertices on the cycle \cite{dfs:t}. That is, all vertices on the cycle are in $\mathit{loop}(u)$.(See Figure \ref{fig:dfs-loops}.)  What makes these statements true is that if vertices are numbered in postorder with respect to the depth-first search that generates the tree, then every arc $(x, y)$ with $x$ numbered less than $y$ is a back arc \cite{st:t}.

\begin{lemma}
\label{lemma:tree-transform}
If $T$ is a depth-first spanning tree and $T'$ is formed from $T$ by contracting $v$ into $p(v)$ or doing $\mathit{transform}(u, v)$, then $T'$ is also a depth-first spanning tree, with preorder on $T'$ being preorder on $T$ restricted to the vertices other than $v$, and postorder on $T'$ being postorder on $T$ restricted to the vertices other than $v$.
\end{lemma}
\begin{proof}
The proof is straightforward.
\end{proof}

We choose $T$ to be a depth-first spanning tree.  To delete vertices on cycles, we add a \textbf{while} loop at the end of the main loop of Algorithm AD that repeatedly does transforms on pairs $u$, $v$ such that $(v, u)$ is a back arc. Just before doing such a transform it adds all vertices in $\mathit{same}(v)$ to $\mathit{same}(u)$. Once there are no such back arcs, any cycle containing $u$ also contains a proper ancestor of $u$. The result is Algorithm GD, which finds dominators in a general graph. It marks arcs just like Version 1 of Algorithm AD.  Each arc added in Step 0 of \emph{transform} is unmarked, and each arc added by a contraction is marked if and only if the arc it replaces was marked.

\begin{figure}[h]
\begin{center}
\fbox{
\begin{minipage}[h]{16cm}
\begin{center}
\textbf{Algorithm GD: Find Dominators in a General Graph, Version 1}
\end{center}
\begin{description}\setlength{\leftmargin}{10pt} \setlength{\itemsep}{0pt} \parskip0pt \parsep0pt
\item[Initialization:] Find a depth-first spanning tree $T$ of the input graph $G$, let $p$ be the parent function of $T$, number the vertices of $T$ from $1$ to $n$ in a bottom-up order, and identify vertices by number.  Unmark all arcs of $G$. Assign $\mathit{same}(u) \leftarrow \{u\}$ for each vertex $u$.\\
\item[Main Loop:]
\textbf{for} $u = 1$ \textbf{until} $n$ \textbf{do} \hfill
\item[] \hspace{16ex} \textbf{while} some tree or forward arc $(u, v)$ is unmarked \textbf{do} \hfill
\item[] \hspace{16ex} $\{$ \ \ mark $(u, v)$; \hfill
\item[] \hspace{19.5ex} \textbf{if} all arcs into $v$ are marked \textbf{then}\hfill
\item[] \hspace{19.5ex} $\{$ \ \ \textbf{if} $u=p(v)$ \textbf{then} \textbf{for} $w \in \mathit{same}(v)$ \textbf{do} $d(w) \leftarrow u$
\item[] \hspace{23.5ex} \textbf{else} $\mathit{same}(p(v)) \leftarrow \mathit{same}(p(v)) \cup \mathit{same}(v)$;
\item[] \hspace{23.5ex} contract $v$ into $p(v)$ \ \ $\}$ \ \ $\}$
\item[] \hspace{16ex} \textbf{while} back or loop arc $(v, u)$ exists \textbf{do} \hfill
\item[] \hspace{19.5ex} \textbf{if} $v = u$ \textbf{then} mark $(v, u)$
\item[] \hspace{19.5ex} \textbf{else} \ \ $\{$ \ \ $\mathit{same}(u) \leftarrow \mathit{same}(u) \cup \mathit{same}(v)$; $\mathit{transform}(u, v)$  \ \ $\}$  \hfill
\end{description}
\end{minipage}
}
\end{center}
\end{figure}

\begin{figure}
\begin{center}
\scalebox{0.7}[0.7]{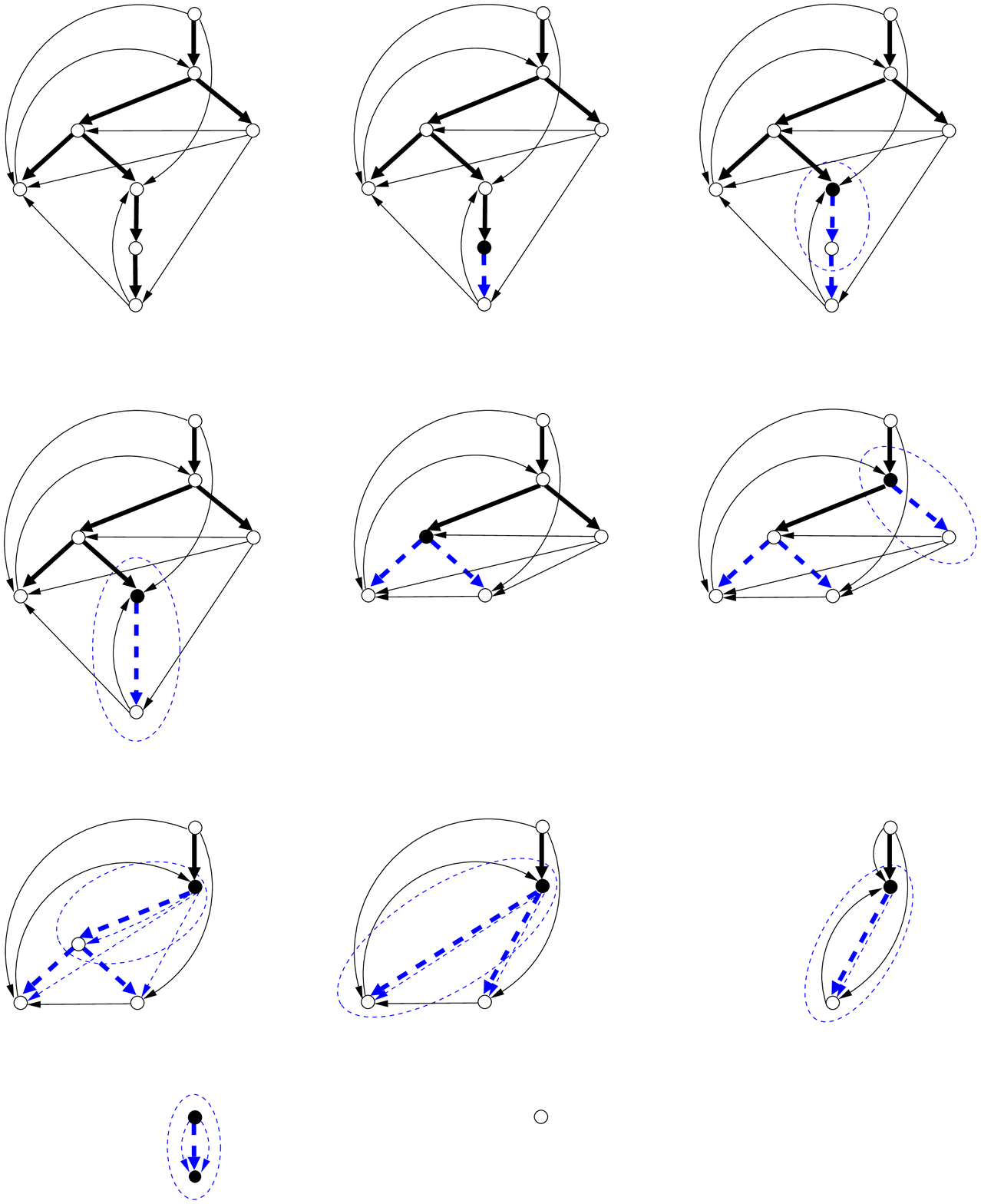}
\end{center}
\caption{\label{fig:general-example} Execution of algorithm GD on the flow graph in Figure \ref{fig:dominator-tree}, using the depth-first spanning tree in Figure \ref{fig:dfs-loops}.
Tree arcs are shown bold; the vertex $u$ processed in the main loop is shown filled; marked arcs are shown dashed; loop arcs created by contractions are omitted.}
\end{figure}

Figure \ref{fig:general-example} illustrates how Algorithm GD algorithm works. The correctness proof of Algorithm GD is analogous to that of Algorithm AD.

\begin{lemma}
\label{lemma:transform}
If $\mathit{transform}(u, v)$ is done during iteration $u$ of the main loop and $p(v) \not= u$, then $\mathit{transform}(u, p(v))$ is also done.
\end{lemma}
\begin{proof}
The operation $\mathit{transform}(u, v)$ adds a back arc $(p(v), u)$.  The existence of such an arc will trigger $\mathit{transform}(u, p(v))$ later in iteration $u$.
\end{proof}

\begin{corollary}
\label{corollary:transform}
If $v$ is deleted during the second \textbf{while} loop of iteration $u$, then so are all ancestors of $v$ that are proper descendants of $u$ in $T$.
\end{corollary}
\begin{proof}
The corollary follows from Lemma \ref{lemma:transform} by induction on the number of ancestors of $v$ that are proper descendants of $u$.
\end{proof}

\begin{lemma}
\label{lemma:general-pass2}
During the main loop of Algorithm GD, except in the middle of the second \textbf{while} loop, no vertex $x < u$ has an entering back arc.
\end{lemma}
\begin{proof}
The lemma is true before the first iteration of the main loop.  Suppose the lemma is true at the beginning of iteration $u$.  A contraction of $v$ into $p(v)$ done during the first \textbf{while} loop of iteration $u$ can only create a new back arc by replacing an arc $(v, z)$ by $(p(v), z)$.  If $(p(v), z)$ is a back arc, so is $(v, z)$.  Since the lemma was true before the contraction, $z \ge u$, which makes the lemma true after the contraction.  By induction on the number of contractions done in the first \textbf{while} loop, the lemma is true at the end of the loop.  Consider the transforms done during the second \textbf{while} loop of iteration $u$.  Step 0 of a transform cannot add any back arcs.  If $\mathit{transform}(u, v)$ adds a back arc $(y, p(v))$, Corollary \ref{corollary:transform} implies that this arc will have been replaced by a loop or back arc $(y', u)$ by the end of the \textbf{while} loop.  The \textbf{while} loop cannot end while a back arc into $u$ exists, so all such arcs become loop arcs by the end of the \textbf{while} loop.  The other possibility is that $\mathit{transform}(u, v)$ adds a back arc $(p(v), z)$.  If $z$ is a descendant of $u$, Corollary \ref{corollary:transform} implies that this back arc is replaced by a loop arc $(u, u)$ by the end of the \textbf{while} loop.  If $z$ is not a descendant of $u$, $z > u$.  We conclude that any back arcs into vertices less than u that are added during the second \textbf{while} loop become loop arcs before the end of the \textbf{while} loop, and any back arcs into $u$ existing before the second \textbf{while} loop or added during the second \textbf{while} loop also become loop arcs before the end of the loop.  Thus the lemma holds at the end of the \textbf{while} loop, and at the beginning of the next iteration of the main loop.  By induction on the number of iterations, the lemma holds throughout the main loop.
\end{proof}

\begin{lemma}
\label{lemma:general-loops}
Throughout the main loop, each loop arc $(x, x)$ has $x \le u$.  During the first \textbf{while} loop, all loop arcs $(x, x)$ with $x < u$ are marked.  At the end of the second \textbf{while} loop, all loop arcs are marked.
\end{lemma}
\begin{proof}
By assumption, the initial graph has no loop arcs.  Each new loop arc is added by a contraction of a vertex $v$ into its parent $p(v)$.  Suppose the lemma holds before the addition.  The new loop arc is $(p(v), p(v))$.  It replaces an arc $(v, v)$, $(p(v), v)$, or $(v, p(v))$.  Furthermore $p(v) \le u$.  Thus the first part of the lemma holds after the addition.  Suppose the new loop is added during the first \textbf{while} loop.  If it replaces an arc $(v, v)$ or $(p(v), v)$, it is marked when added.  If it replaces $(v, p(v))$, Lemma \ref{lemma:general-pass2} implies that $p(v) = u$.  Thus the second part of the lemma holds after the addition.

Suppose the new loop arc is added during the second \textbf{while} loop.  Corollary \ref{corollary:transform} implies that either the new loop arc is $(u, u)$ or it will be replaced by $(u, u)$ before the second \textbf{while} loop ends.  It follows that the second \textbf{while} loop cannot end while there is an unmarked loop $(x, x)$ with $x < u$.  It also cannot end while there is an unmarked loop $(u, u)$.
\end{proof}

\begin{lemma}
\label{lemma:general-marked}
Throughout the main loop, each marked arc $(x, y)$ is a tree, forward, or loop arc such that $x \le u$.
\end{lemma}
\begin{proof}
Each arc newly marked by the algorithm satisfies the lemma, and contractions and transforms preserve it.
\end{proof}

\begin{lemma}
\label{lemma:invariants-2}
During any iteration of the \textbf{for} loop of the main loop, every undeleted vertex $v \not= s$ has $d(v) \ge u$; during the second \textbf{while} loop of the iteration, every undeleted vertex $v \not= s$ has $d(v) > u$.  Each contraction done during the first \textbf{while} loop is such that $u$ and $v$ satisfy the hypotheses of Lemmas \ref{lemma:parent-dom} and \ref{lemma:contract-1}.  Each operation $\mathit{transform}(u, v)$ is such that $u$ and $v$ satisfy the hypotheses of Lemmas \ref{lemma:loop} and \ref{lemma:loop-transform}.
\end{lemma}
\begin{proof}
By the proof of Lemma \ref{lemma:invariants}, (i) the first invariant holds at the beginning of the first iteration; (ii) if the first invariant holds at the beginning of the first iteration, it continues to hold throughout the first \textbf{while} loop; and (iii) by Lemma \ref{lemma:general-loops}, if the first invariant holds before a contraction done during the first \textbf{while} loop, $u$ and $v$ satisfy the hypotheses of Lemmas \ref{lemma:parent-dom} and \ref{lemma:contract-1}.  The proof of Lemma \ref{lemma:invariants} also implies that if the first invariant holds at the end of the first \textbf{while} loop, then so does the second, since the proof that there is no undeleted vertex $v$ such that $d(v) = u$ only requires that there are no cycles among the proper descendants of $u$, which follows from Lemma \ref{lemma:general-pass2}.  If the second invariant holds before an operation $\mathit{transform}(u, v)$, then $u$ and $v$ satisfy the hypotheses of Lemmas \ref{lemma:loop} and \ref{lemma:loop-transform}, and by Lemma \ref{lemma:loop-transform} the second invariant holds after the transform.  The lemma follows by induction.
\end{proof}

\begin{corollary}
\label{corollary:general-start-vertex}
At the end of the main loop, only $s$ is undeleted.
\end{corollary}
\begin{proof}
Once the second invariant of Lemma \ref{lemma:invariants-2} holds for $u = n$, $s$ is the only undeleted vertex.
\end{proof}

\begin{theorem}
\label{theorem:GD-correct}
Algorithm GD is correct.
\end{theorem}
\begin{proof}
The theorem is immediate from Lemmas \ref{lemma:loop}, \ref{lemma:loop-transform}, \ref{lemma:invariants-2}, and Corollary \ref{corollary:general-start-vertex}.
\end{proof}

We make Algorithm GD efficient in the same way as Algorithm AD.  First we observe that the algorithm remains correct with Step 0 of $\mathit{transform}$ deleted, so that it becomes merely a contraction.  Furthermore the algorithm remains correct if all marked arcs are deleted.  Arcs that are marked when added have no effect on the computation.  A new arc $(p(u), w)$ added by $\mathit{transform}(u, v)$ is unmarked when added, but if Algorithm GD marks each such arc at the beginning of iteration $p(u)$ of the main loop, the effect is the same as if the arc had never been added in the first place.  Once an arc is marked, it has no further effect on the behavior of the algorithm.

\begin{theorem}
\label{theorem:GD-correct-2}
Algorithm GD with each $\mathit{transform}$ replaced by a contraction and with all marked arcs deleted is correct.
\end{theorem}
\begin{proof}
Consider a run of Algorithm GD with transforms replaced by contractions and arcs deleted when they are marked.  Rerun the original version of Algorithm GD, choosing the same spanning tree and the same vertex numbering, preferentially marking arcs added by Step 0 of \emph{transform}, and otherwise making the same choices of arcs to mark as in the run with the transforms replaced and marked arcs deleted.  An induction on the number of arcs marked shows that this is always possible.  The two runs compute the same immediate dominators.
\end{proof}

As in Algorithm AD, we use a disjoint set data structure to keep track of contractions.  We maintain the original arc ends and use the disjoint set structure to map these to the current ends: if $(x, y)$ is an original arc, the corresponding current arc is $(\mathit{find}(x), \mathit{find}(y))$.  In the main loop we add arcs to the graph at the same time as in Algorithm AD.  Specifically, if $(x, y)$ is an arc such that $\mathit{nca}(x, y) = u$, we add $(x, y)$ to the graph at the beginning of iteration $u$. Arc $(x, y)$ cannot become a tree or forward arc until $x$ is replaced by an ancestor of $u$, which cannot happen until iteration $u$.  Similarly, arc $(x, y)$ cannot become a back or loop arc until $y$ is replaced by an ancestor of $u$. We need to keep track of arcs both into and out of vertices, so we maintain two bags of vertices for each undeleted vertex $v$: $\mathit{in}(v)$, which contains a copy of vertex $x$ for each original arc $(x, y)$ such that $v = \mathit{find}(y)$ and $(x, y)$ has been added to the graph, and $\mathit{out}(v)$, which contains a copy of vertex $y$ for each original arc $(x, y)$ such that $v = \mathit{find}(x)$ and $(x, y)$ has been added to the graph.  As arcs are marked, the corresponding entries are deleted from \emph{in}-bags and \emph{out}-bags, but some of these deletions are done lazily, as we discuss below.  If $y$ is in $\mathit{out}(u)$ during iteration $u$, then there is a corresponding tree, forward, or loop arc $(u, \mathit{find}(y))$; if $x$ is in $\mathit{in}(u)$ during iteration $u$, then there is a corresponding back or loop arc $(\mathit{find}(x), u)$.

Because of the lazy deletions, we need a slightly more complicated way of deciding when to do contractions in the first \textbf{while} loop than in Algorithm AD.  In addition to maintaining a count $\mathit{total}(v)$ of the total number of unmarked arcs into vertex $v$, we maintain a count $\mathit{added}(v)$ of the number of
arcs into $v$ that have been added to the graph during the main loop but are not yet marked.  When a tree or forward arc $(u, v)$ is marked in the first \textbf{while} loop, we decrement both $\mathit{total}(v)$ and $\mathit{added}(v)$.  To avoid the possibility of marking loop arcs twice (once in each \textbf{while} loop), we mark loop arcs only in the second \textbf{while} loop, and we do the marking implicitly at the end of this \textbf{while} loop.  At the end of iteration $u$, all added arcs into $u$ are loop arcs.  Some of these may have been marked earlier, when they were tree or forward arcs.  The value of $\mathit{added}(u)$
is the number of such arcs that have not yet been marked.  To mark these implicitly, we assign $\mathit{total}(u) \leftarrow \mathit{total}(u) - \mathit{added}(u)$ and $\mathit{added}(u) \leftarrow 0$.  This avoids the need for synchronized deletions of corresponding vertices from \emph{in}-bags and \emph{out}-bags, allowing some of these deletions to be done late or not at all.


\begin{figure}
\begin{center}
\fbox{
\begin{minipage}[h]{16cm}
\begin{center}
\textbf{Algorithm GD: Find Dominators in a General Graph, Version 2}
\end{center}
\begin{description}\setlength{\leftmargin}{10pt} \setlength{\itemsep}{0pt} \parskip0pt \parsep0pt
\item[Initialization:] Find a depth-first spanning tree of the input graph, let $p$ be its parent function, number its vertices in a bottom-up order, and identify vertices by number. Compute $\mathit{total}(u) \leftarrow |\{(x, u) \in A\}|$ and $\mathit{arcs}(u) \leftarrow \{(x, y) \in A \ | \ \mathit{nca}(x, y) = u\}$ for each vertex $u$.\\
\item[Main Loop:]
\textbf{for} $u = 1$ \textbf{until} $n$ \textbf{do} \hfill
\item[] \hspace{12.5ex}  $\{$ \ \ $\mathit{out}(u) \leftarrow [ \ ]$; $\mathit{in}(u) \leftarrow [ \ ]$;  \hfill
\item[] \hspace{16ex} $\mathit{make\text{-}set}(u)$; $\mathit{added}(u) \leftarrow 0$; \hfill
\item[] \hspace{16ex} $\mathit{same}(u) \leftarrow \{u\}$;
\item[] \hspace{16ex} \textbf{for} $(x, y) \in \mathit{arcs}(u)$ \textbf{do}\hfill
\item[] \hspace{16ex} $\{$ \ \ add $y$ to $\mathit{out}(\mathit{find}(x))$; add $x$ to $\mathit{in}(\mathit{find}(y))$;
\item[] \hspace{19.5ex} $\mathit{added}(\mathit{find}(y)) \leftarrow \mathit{added}(\mathit{find}(y))$ + 1 \ \ $\}$ \hfill
\item[] \hspace{16ex} \textbf{while} $\mathit{out}(u) \not= [ \ ]$ \textbf{do} \hfill
\item[] \hspace{16ex} $\{$ \ \ delete some $y$ from $\mathit{out}(u)$;  $v \leftarrow \mathit{find}(y)$; \hfill
\item[] \hspace{19.5ex} \textbf{if} $v \not= u$ \textbf{then}
\item[] \hspace{19.5ex} $\{$ \ \ $\mathit{total}(v) \leftarrow \mathit{total}(v) - 1$; $\mathit{added}(v) \leftarrow \mathit{added}(v) - 1$ \ \ $\}$ \hfill
\item[] \hspace{19.5ex} \textbf{if} $\mathit{total}(v) = 0$ \textbf{then} \hfill
\item[] \hspace{19.5ex} $\{$ \ \ $x \leftarrow \mathit{find}(p(v))$; \hfill
\item[] \hspace{23ex} \textbf{if} $u = x$ \textbf{then} \textbf{for} $w \in \mathit{same}(v)$ do $d(w) \leftarrow u$
\item[] \hspace{23ex} \textbf{else} $\mathit{same}(x) \leftarrow \mathit{same}(x) \cup \mathit{same}(v)$; \hfill
\item[] \hspace{23ex} $\mathit{unite}(p(v), v)$; $\mathit{out}(x) \leftarrow \mathit{out}(x) \cup \mathit{out}(v)$ \ \ $\}$ \ \ $\}$ \hfill
\item[] \hspace{16ex} \textbf{while} $\mathit{in}(u) \not= [ \ ]$ \textbf{do} \hfill
\item[] \hspace{16ex} $\{$ \ \ delete some $z$ from $\mathit{in}(u)$;  $v \leftarrow \mathit{find}(z)$; \hfill
\item[] \hspace{19.5ex} \textbf{while} $v \not= u$ \textbf{do} \hfill
\item[] \hspace{19ex} \ $\{$ \ \ $\mathit{same}(u) \leftarrow \mathit{same}(u) \cup \mathit{same}(v)$; \hfill
\item[] \hspace{23ex} $x \leftarrow \mathit{find}(p(v))$;  $\mathit{unite}(p(v), v)$; \hfill
\item[] \hspace{23ex} $\mathit{in}(x) \leftarrow \mathit{in}(x) \cup \mathit{in}(v)$; \hfill
\item[] \hspace{23ex} $\mathit{out}(x) \leftarrow \mathit{out}(x) \cup \mathit{out}(v)$; \hfill
\item[] \hspace{23ex} $\mathit{total}(x) \leftarrow \mathit{total}(x) + \mathit{total}(v)$; \hfill
\item[] \hspace{23ex} $\mathit{added}(x) \leftarrow \mathit{added}(x) + \mathit{added}(v)$;
\item[] \hspace{23ex} $v \leftarrow x$ \ \ $\}$ \ \ $\}$ \hfill
\item[] \hspace{16ex} $\mathit{total}(u) \leftarrow \mathit{total}(u) - \mathit{added}(u)$; $\mathit{added}(u) \leftarrow 0$ \ \ $\}$ \hfill \\
\end{description}
\end{minipage}
}
\end{center}
\end{figure}

The resulting implementation is Version 2 of Algorithm GD. The body of the loop ``\textbf{while} $v \not= u$ \textbf{do}'' inside the second \textbf{while} loop of the main loop adds the vertices in $\mathit{same}(v)$ to $\mathit{same}(u)$, contracts $v$ into its parent $x$, and replaces $v$ by $x$. The effect of this loop is to contract the entire tree path from $u$ to the original $v$ into $u$.  Instead of waiting to delete a back arc $(v, u)$ until it becomes a loop arc, the algorithm deletes it as soon as it is found but does contractions as if the arc were undeleted until it becomes a loop arc. Whereas a contraction in the first \textbf{while} loop can change only the first ends of arcs and does not add any new arcs into a vertex (since $v$ has no incoming arcs when contracted), a contraction in the second \textbf{while} loop can change either end of an arc, and can add new arcs into $x$. Thus not only $\mathit{in}(x)$ but also $\mathit{out}(x)$, $\mathit{total}(x)$, and $\mathit{added}(x)$ must be updated.

\begin{lemma}
\label{lemma:general-invariants2-a}
During either \textbf{while} loop in iteration $u$ of the main loop, for each original arc $(x, y)$ such that $\mathit{nca}(x, y) \le u$ and whose corresponding arc is unmarked, $y$ is in $\mathit{out}(\mathit{find}(x))$ and $x$ is in $\mathit{in}(\mathit{find}(y))$.
\end{lemma}
\begin{proof}
The proof is straightforward by induction on the number of steps done by the main loop.
\end{proof}

\begin{lemma}
\label{lemma:general-invariants2-b}
Let $(x, y)$ be an original arc that is a back arc or becomes a back arc. Let $u$ be the iteration of the main loop during which $(x, y)$ is added if it is originally a back arc or during which it becomes a back arc.  During the second \textbf{while} loop of iteration $u$, vertex $x$ will be deleted from $\mathit{in}(u)$ and $(x, y)$ will become a loop arc $(u, u)$.  Vertex $y$ will never be deleted from an $\mathit{out}$-set.
\end{lemma}
\begin{proof}
The proof is straightforward by induction on the number of steps taken by the main loop.
\end{proof}

\begin{lemma}
\label{lemma:general-invariants2-c}
Let $(x, y)$ be an original arc that is a tree or forward arc or becomes a tree or forward arc. If $x$ is deleted from $\mathit{in}(u)$ during the second \textbf{while} loop of iteration $u$ of the main loop, $(x, y)$ has become a loop arc; that is, $\mathit{find}(x) = u$.
\end{lemma}
\begin{proof}
The proof is straightforward by induction on the number of steps taken by the main loop.
\end{proof}

\begin{lemma}
\label{lemma:general-invariants2-d}
The main loop correctly maintains $\mathit{total}$ and $\mathit{added}$.
\end{lemma}
\begin{proof}
Mark an original arc $(z, y)$ when $y$ is deleted from some $\mathit{out}(u)$ such that $\mathit{find}(y) \not= u$ or at the end of the second \textbf{while} loop during which $z$ is deleted from some $\mathit{in}(u)$, whichever comes first.  When the former happens, $(z, y)$ is a tree or forward arc.  With this definition, the main loop maintains the invariant that, for each vertex $v$, $\mathit{total}(v)$ is the number of unmarked arcs into $v$ and $\mathit{added}(v)$ is the number of unmarked arcs into $v$ so far added by the inner \textbf{for} loop, as can be proved by induction on the number of steps taken by the main loop.  Also, the proof of Lemma \ref{lemma:general-loops} shows that after each iteration of the main loop, all loop arcs are marked by this definition.   \end{proof}

\begin{theorem}
\label{lemma:GD2-correct}
Version 2 of Algorithm GD is correct.
\end{theorem}
\begin{proof}
The theorem is immediate from Theorem \ref{theorem:GD-correct-2} and Lemmas \ref{lemma:general-invariants2-a}, \ref{lemma:general-invariants2-b}, \ref{lemma:general-invariants2-c}, and \ref{lemma:general-invariants2-d}.
\end{proof}

Version 2 of Algorithm GD has a running time of $O(m \alpha(n, m/n))$, or $O(m)$ on a RAM, depending on the disjoint set implementation.  The only data structure needed other than simple lists and maps is one to maintain disjoint sets.

Unlike Algorithm AD, Algorithm GD cannot avoid the need to compute nearest common ancestors to determine when to add vertices to \emph{in}-bags and \emph{out}-bags.  If the vertex order is postorder, nca's are not needed for filling \emph{in}-bags; if the order is reverse preorder, nca's are not needed for filling \emph{out}-bags; but no order suffices for both. The crucial difficulty is that one cannot tell a priori whether a cross arc will become a back arc or a tree or forward arc. The algorithm presented in the conference version of our paper \cite{dominators:FM12} does not compute nca's, and consequently is incorrect on general graphs, although it is correct for acyclic graphs.

Another way to keep track of added arcs that is perhaps more intuitive than the one used in Version 2 of Algorithm GD is to maintain the added arcs themselves instead of their ends.  Then $\mathit{in}(v)$ and $\mathit{out}(v)$ become sets of arcs rather than bags of vertices.  When marking an arc, we delete it from both sets containing it.  We can avoid the need to implement these sets as doubly-linked lists (needed to support arbitrary deletion) by marking arcs as deleted.  This approach avoids the need to maintain \emph{added} counts as well as \emph{total} counts.  But the approach we have chosen has the advantage that it allows us to separate the acyclic part of the algorithm from the cyclic part.  We shall exploit this ability in Section \ref{sec:loops}.

Like Algorithm AD, Algorithm GD can be modified to build the spanning tree and compute dominators in a single depth-first search.  We conclude this section by presenting Version 3 of Algorithm GD, which does this.  The algorithm uses two disjoint set data structures concurrently: one to keep track of contractions, the other to compute nca's.  We use the prefix ``$n$-'' to indicate the operations of the latter.  To compute nca's, the algorithm uses the following streamlined version \cite{pathcomp:t} of the nca algorithm of Hopcroft and Ullman \cite{lca:ahu}:  During the depth-first search, when the algorithm first visits a vertex $u$ (the preorder visit to $u$), it does \emph{n-makeset}$(u)$.  When retreating along a tree arc $(p(u), u)$, it does \emph{unite}$(p(u), u)$.  This maintains the invariant that if $u$ is the current vertex of the search and $v$ is any visited vertex, \emph{n-find}$(v)$ is the nearest common ancestor of $u$ and $v$ in the depth-first spanning tree.  Thus when the search retreats along an arc $(u, v)$, \emph{nca}$(u, v) = $ \emph{n-find}$(v)$.

To run the depth-first search, the algorithm maintains a bit $\mathit{visited}(u)$ for each vertex $u$ such that $\emph{visited}(u)$ is \textbf{true} if and only if $u$ has been visited in preorder. When visiting a vertex $u$ in postorder, the algorithm executes iteration $u$ of the main loop.


\begin{figure}
\begin{center}
\fbox{
\begin{minipage}[h]{16cm}
\begin{center}
\textbf{Algorithm GD: Find Dominators in a General Graph, Version 3}
\end{center}
\begin{description}\setlength{\leftmargin}{10pt} \setlength{\itemsep}{0pt} \parskip0pt \parsep0pt
\item[] \textbf{for} $u \in V$ \textbf{do} $\mathit{visited}(u) \leftarrow \textbf{false}$; \hfill
\item[] $\mathit{dfs}(s)$; \hfill\\
\item[] procedure $\mathit{dfs}(u)$:
\item[] $\{$ \ \ $\mathit{previsit}(u)$;
\item[] \hspace{3ex} \textbf{for} $(u, v) \in A$ \textbf{do}
\item[] \hspace{3ex} $\{$ \ \ \textbf{if} $\mathit{visited}(v) = \mathbf{false}$ \textbf{then} \ $\{$ $\mathit{dfs}(v)$; $p(v) \leftarrow u$; $\mathit{n\text{-}unite}(p(v), v)$ $\}$
\item[] \hspace{3ex} \ \ \ \ $\mathit{total}(v) \leftarrow \mathit{total}(v) + 1$; add $(u, v)$ to $\mathit{arcs}(\mathit{n\text{-}find}(v))$ \ \ $\}$
\item[] \ \ \ \ \ $\mathit{postvisit}(u)$ \ $\}$\\
\item[] procedure $\mathit{previsit}(u)$:
\item[] $\{$ \ \ $\mathit{visited}(u) \leftarrow \textbf{true}$; $\mathit{total}(u) \leftarrow 0$; $\mathit{arcs}(u) \leftarrow [ \ ]$; $\mathit{n\text{-}makeset}(u)$ \ \ $\}$\\
\item[] procedure $\mathit{postvisit}(u)$:
\item[] $\{$ \ \ $\mathit{out}(u) \leftarrow [\ ]$; $\mathit{in}(u) \leftarrow [\ ]$;
\item[] \hspace{3ex} $\mathit{make\text{-}set}(u)$; $\mathit{added}(u) \leftarrow 0$;
\item[] \hspace{3ex} $\mathit{same}(u) \leftarrow \{u\}$;
\item[] \hspace{3ex} \textbf{for} $(x, y) \in \mathit{arcs}(u)$ \textbf{do}
\item[] \hspace{3ex} $\{$ \ add $y$ to $\mathit{out}(\mathit{find}(x))$; add $x$ to $\mathit{in}(\mathit{find}(y))$;
\item[] \hspace{3ex} \ \ \ \ $\mathit{added}(\mathit{find}(y)) \leftarrow \mathit{added}(\mathit{find}(y)) + 1$ \ $\}$
\item[] \hspace{3ex} \textbf{while} $\mathit{out}(u) \not= [\ ]$ \textbf{do}
\item[] \hspace{3ex} $\{$ \ delete some $y$ from $\mathit{out}(u)$; $v \leftarrow \mathit{find}(y)$;
\item[] \hspace{6ex} \textbf{if} $v \not= u$ \textbf{then} $\{$ $\mathit{total}(v) \leftarrow \mathit{total}(v) - 1$; $\mathit{added}(v) \leftarrow \mathit{added}(v) - 1$ $\}$
\item[] \hspace{6ex} \textbf{if} $\mathit{total}(v) = 0$ \textbf{then}
\item[] \hspace{6ex} $\{$ \ $x \leftarrow \mathit{find}(p(v))$;
\item[] \hspace{6ex} \ \ \ \ \textbf{if} $u=x$ \textbf{then} \textbf{for} $w \in \mathit{same}(v)$ \textbf{do} $d(w) \leftarrow v$
\item[] \hspace{6ex} \ \ \ \ \textbf{else} $\mathit{same}(x) \leftarrow \mathit{same}(x) \cup \mathit{same}(v)$;
\item[] \hspace{6ex} \ \ \ \ $\mathit{unite}(p(v), v)$; $\mathit{out}(x) \leftarrow \mathit{out}(x) \cup \mathit{out}(v)$ \ $\}$ \ $\}$
\item[] \hspace{3ex} \textbf{while} $\mathit{in}(u) \not= [\ ]$ \textbf{do}
\item[] \hspace{3ex} $\{$ \ delete some $z$ from $\mathit{in}(u)$; $v \leftarrow \mathit{find}(z)$;
\item[] \hspace{6ex} \textbf{while} $v \not= u$ \textbf{then}
\item[] \hspace{6ex} $\{$ \ $\mathit{same}(u) \leftarrow \mathit{same}(u) \cup \mathit{same}(v)$
\item[] \hspace{9ex} $x \leftarrow \mathit{find}(p(v))$; $\mathit{unite}(p(v), v)$;
\item[] \hspace{9ex} $\mathit{in}(x) \leftarrow \mathit{in}(x) \cup \mathit{in}(v)$;
\item[] \hspace{9ex} $\mathit{out}(x) \leftarrow \mathit{out}(x) \cup \mathit{out}(v)$;
\item[] \hspace{9ex} $\mathit{total}(x) \leftarrow \mathit{total}(x) + \mathit{total}(v)$;
\item[] \hspace{9ex} $\mathit{added}(x) \leftarrow \mathit{added}(x) + \mathit{added}(v)$;
\item[] \hspace{9ex} $v \leftarrow x$ \ $\}$ \ $\}$
\item[] \ \ \ \ \ $\mathit{total}(u) \leftarrow \mathit{total}(u) - \mathit{added}(u)$; $\mathit{added}(u) \leftarrow 0$ \ $\}$
\end{description}
\end{minipage}
}
\end{center}
\end{figure}


\begin{theorem}
\label{theorem:GD3-correct}
Version 3 of Algorithm GD is correct.
\end{theorem}
\begin{proof}
The Theorem follows from Theorem \ref{lemma:GD2-correct} and the correctness of the nca computation.
\end{proof}

\section{Loops and Dominators}
\label{sec:loops}

Many applications of dominators to global flow analysis also require a loop nesting forest $H$ (defined with respect to some depth-first search tree $T$).  We can find such an $H$ as part of the computation of dominators.  Indeed, the part of Algorithm GD that does transforms is almost the same as Tarjan's algorithm for finding $H$, and if we eliminate the other contractions done by Algorithm GD, we obtain a version of Tarjan's algorithm.  We can turn this idea around: we find $H$ using Tarjan's algorithm, and then find dominators by running Algorithm GD modified to choose transforms using $H$.
This gives us Algorithm HD, which computes $H$ as well as $D$.  The algorithm consists of three steps: initialization, the $H$ loop, which computes $H$, and the $D$ loop, which computes $D$.  The $D$ loop is the main loop of Algorithm GD, modified to use $H$ to select transforms and to process vertices in reverse preorder. The use of preorder eliminates the need to compute nca's; the use of $H$ eliminates the need to use \emph{in}-bags.  The $H$ loop is a depth-first search that builds a depth-first spanning tree $T$ with parent function $p$, constructs a list $\mathit{revpre}$ of the vertices in reverse preorder with respect to the search for use in the $D$ loop, builds the loop nesting forest $H$ with respect to $T$, and finds an exit $\mathit{exit}(v)$ from each loop $\mathit{loop}(v)$ that is contained in a larger loop.  To build $H$, it processes each vertex $u$ in postorder, doing a backward search from $u$ among its descendants in $T$.  Any bottom-up order will do, but using postorder avoids the need to compute nca's to determine which arcs to traverse in backward searches \cite{dominators:bgkrtw}.  The vertices reached by the backward search are exactly those in $\mathit{loop}(u)$. To avoid revisiting vertices that are in nested loops, the algorithm contracts each loop into its head once it is found. This is Tarjan's loop-finding algorithm~\cite{st:t} streamlined to eliminate the need for nca's \cite{dominators:bgkrtw}. The computation of $H$ is essentially the same as Algorithm GD with the contractions not in transforms and the computation of $D$ omitted. The $H$ and $D$ loops can use the same disjoint set data structure; the $D$ loop reinitializes the sets to be singletons.

From $H$, a set of loop exits, and $D$, one can compute a low-high order of the vertices, which suffices to easily verify that $D$ is correct \cite{domcert}.  Computing a low-high order takes an extra $O(m)$-time pass, which we omit; \cite{domcert} gives the details of this computation.  If a low-high order is not needed, loop exits need not be computed, and \emph{in}-bags of vertices (as in Algorithm GD) can be used in place of arc sets ($\mathit{in\text{-}arcs(v)}$) to guide the backward searches. The $H$ loop can be extended to test whether the graph is reducible and to identify the reducible loops \cite{loops:Ramalingam1999,st:t}.


\begin{figure}
\begin{center}
\fbox{
\begin{minipage}[h]{16cm}
\begin{center}
\textbf{Algorithm HD: Find Loops and Dominators in a General Graph}
\end{center}
\begin{description}\setlength{\leftmargin}{10pt} \setlength{\itemsep}{0pt} \parskip0pt \parsep0pt
\item[Initialization:] \textbf{for} $u$ in $V$ \textbf{do} $\mathit{visited}(u) \leftarrow \textbf{false}$; $\mathit{revpre} = [\ ]$; $H \leftarrow \{ \ \}$;
\item[$H$ loop:] $\mathit{dfs}(s)$;
\item[$D$ loop:] \textbf{for} $u \in \mathit{revpre}$ \textbf{do}
\item[] \hspace{7.5ex} $\{$ \; $\mathit{out}(u) \leftarrow [\ ]$; $\mathit{make\text{-}set}(u)$; $\mathit{added}(u) \leftarrow 0$; $\mathit{same}(u) \leftarrow \{u\}$;
\item[] \hspace{10.5ex} \textbf{for} $(x, u) \in A$ \textbf{do}
\item[] \hspace{10.5ex} $\{$ \ add $u$ to $\mathit{out}(\mathit{find}(x))$; $\mathit{added}(\mathit{find}(u)) \leftarrow \mathit{added}(\mathit{find}(u)) + 1$ \ $\}$
\item[] \hspace{10.5ex} \textbf{while} $\mathit{out}(u) \not= [\ ]$ \textbf{do}
\item[] \hspace{10.5ex} $\{$ \ delete some $y$ from $\mathit{out}(u)$; $v \leftarrow \mathit{find}(y)$;
\item[] \hspace{13.5ex} \textbf{if} $v \not= u$ \textbf{then} $\{$ $\mathit{total}(v) \leftarrow \mathit{total}(v) - 1$; $\mathit{added}(v) \leftarrow \mathit{added}(v) - 1$ $\}$
\item[] \hspace{13.5ex} \textbf{if} $\mathit{total}(v) = 0$ \textbf{then}
\item[] \hspace{13.5ex} $\{$ \ $x \leftarrow \mathit{find}(p(v))$;
\item[] \hspace{16.5ex} \textbf{if} $u = x$ \textbf{then} \textbf{for} $w \in \mathit{same}(v)$ do $d(w) \leftarrow u$
\item[] \hspace{16.5ex} \textbf{else} $\mathit{same}(x) \leftarrow \mathit{same}(x) \cup \mathit{same}(v)$;
\item[] \hspace{16.5ex} $\mathit{unite}(p(v), v)$; $\mathit{out}(x) \leftarrow \mathit{out}(x) \cup \mathit{out}(v)$ \ $\}$ \ $\}$
\item[] \hspace{10.5ex} \textbf{for} $(u, z) \in H$ \textbf{do}
\item[] \hspace{10.5ex} $\{$ \ $v \leftarrow \mathit{find}(z)$;
\item[] \hspace{13.5ex} \textbf{if} $v \not= u$ \textbf{then}
\item[] \hspace{13.5ex} $\{$ \ $\mathit{same}(u) \leftarrow \mathit{same}(u) \cup \mathit{same}(v)$;
\item[] \hspace{16.5ex} $x \leftarrow \mathit{find}(p(v))$; $\mathit{unite}(p(v), v)$;
\item[] \hspace{16.5ex} $\mathit{out}(x) \leftarrow \mathit{out}(x) \cup \mathit{out}(v)$;
\item[] \hspace{16.5ex} $\mathit{total}(x) \leftarrow \mathit{total}(x) + \mathit{total}(v)$;
\item[] \hspace{16.5ex} $\mathit{added}(x) \leftarrow \mathit{added}(x) + \mathit{added}(v)$ \ $\}$ \ $\}$
\item[] \hspace{10.5ex} $\mathit{total}(u) \leftarrow \mathit{total}(u) - \mathit{added}(u)$; $\mathit{added}(u) \leftarrow 0$ \ $\}$\\
\item[] procedure $\mathit{dfs}(u)$:
\item[] $\{$ \ \ $\mathit{previsit}(u)$;
\item[] \hspace{3ex} \textbf{for} $(u, v) \in A$ \textbf{do}
\item[] \hspace{3ex} $\{$ \ \textbf{if} $\mathit{visited}(v) = \textbf{false}$ \textbf{then} $\{$ $\mathit{dfs}(v)$; $p(v) \leftarrow u$ $\}$
\item[] \hspace{6ex} $\mathit{total}(v) \leftarrow \mathit{total}(v) + 1$; add $(u, v)$ to $\mathit{in\text{-}arcs}(\mathit{find}(v))$ \ $\}$
\item[] \hspace{3ex} $\mathit{postvisit}(u)$ \ $\}$\\
\item[] procedure $\mathit{previsit}(u)$:
\item[] $\{$ \ \ $\mathit{visited}(u) \leftarrow \textbf{true}$; $\mathit{total}(u) \leftarrow 0$; $\mathit{in\text{-}arcs}(u) \leftarrow [\ ]$;
\item[] \hspace{3ex} $\mathit{make\text{-}set}(u)$; add $u$ to the front of $\mathit{revpre}$ \ $\}$\\
\item[] procedure $\mathit{postvisit}(u)$:
\item[] $\{$ \ \ \textbf{while} $\mathit{in\text{-}arcs}(u) \not= [\ ]$ \textbf{do}
\item[] \hspace{3ex}  $\{$ \ delete some $(z, y)$ from $\mathit{in\text{-}arcs}(u)$;
\item[] \hspace{6ex}  $v \leftarrow \mathit{find}(z)$;
\item[] \hspace{6ex}  \textbf{while} $v \not= u$ \textbf{do}
\item[] \hspace{6ex}  $\{$ \ add $(u, v)$ to $H$; $\mathit{exit}(v) \leftarrow (z, y)$;
\item[] \hspace{9ex}  $x \leftarrow \mathit{find}(p(v))$; $\mathit{unite}(p(v), v)$;
\item[] \hspace{9ex}  $\mathit{in}(x) \leftarrow \mathit{in}(x) \cup \mathit{in}(v)$;
\item[] \hspace{9ex}  $v \leftarrow x$ \ $\}$ \ $\}$ \ $\}$
\end{description}
\end{minipage}
}
\end{center}
\end{figure}


\begin{theorem}
\label{LD-correct}
Algorithm HD is correct.
\end{theorem}
\begin{proof}
The theorem is immediate from the proof of Theorem \ref{lemma:GD2-correct} and the correctness of the $H$ loop.
\end{proof}

The running time of Algorithm HD is $O(m \alpha(n, m/n))$, or $O(m)$ on a RAM, depending on the disjoint set implementation.  The $D$ loop can be folded into the $H$ loop, but only at the cost of doing an nca computation to fill the \emph{in}-bags.  The result is an algorithm that consists of initialization and a single depth-first search, but that uses three disjoint set data structures concurrently, one to find nca's, one to find loops, and one to find dominators. Keeping the $H$ and $D$ loops separate seems likely to gives a more efficient implementation, but this is a question to be resolved by experiments.

Ramalingam's reduction uses a loop nesting forest in a different way: to transform the graph by deleting certain arcs and adding other arcs and vertices to make the graph acyclic while preserving dominators. A streamlined version of his transformation that avoids adding vertices is the following. Let $G$ be a flow graph, let $T$ be a depth-first spanning tree of $G$ with parent function $p$, and let $H$ be the loop nesting forest of $G$ with respect to $T$. Form $G'$ from $G$ by (i) deleting all back arcs; (ii) for each $u$ such that $\mathit{loop}(u)$ and $\mathit{loop}(h(u))$ have a common non-head entry $(v, w)$, add arc $(p(h(u)), u)$; and (ii) for each non-head entry $(v, w)$ (into some loop), add arc $(v, u)$, where $\mathit{loop}(u)$ is the largest loop with entry $(v, w)$.  Since all loops containing $w$ are nested, $u$ is an ancestor in $T$ of every $u'$ such that $(v, w)$ enters $\mathit{loop}(u')$, which implies that $(v, w)$ is a non-head entry of $\mathit{loop}(u)$.  Every arc added by (ii) is a forward arc, and every arc added by (iii) is a forward or cross arc, so $G'$ is acyclic, since it contains no back arcs.  The number of arcs added by (ii) is at most $n - 2$ and by (iii) at most $m$, so $G'$ has $O(m)$ arcs.

\begin{theorem}
\label{theorem:Ramalingam-reduction}
Graphs $G$ and $G'$ have the same dominators.
\end{theorem}
\begin{proof}
Suppose $x$ does not dominate $y$ in $G$.  Then there is a simple path $P$ from $s$ to $y$ in $G$ that avoids $x$.
This path is in $G'$ unless it contains a back arc. If it does, let $(z, u')$ be the last back arc on $P$. Since $P$ is simple, $(z, u')$ is preceded on $P$ by a non-head entry $(v, w)$ to $\mathit{loop}(u')$.  Let $\mathit{loop}(u)$ be the largest loop entered by $(v, w)$.  Then $u$ is an ancestor of $u'$ in $T$.  Also $x \not= u'$.
We modify $P$ to produce a path in $G'$ from $s$ to $y$ that avoids $x$.

There are three cases.  If $x$ is not a proper ancestor of $u'$ in $T$, then replacing the part of $P$ from $s$ to $u'$ by the path in $T$ from $s$ to $u'$ gives the desired path. If $x$ is a proper ancestor of $u$ in $T$, then replacing the part of $P$ from $v$ to $u'$ by the new arc $(v, u)$ followed by the path in $T$ from $u$ to $u'$ gives the desired path.  The third and last case is $x$ a descendant of $u$ and a proper ancestor of $u'$ in $T$.  Let $u''$ be the ancestor of $u'$ such that $x$ is a proper ancestor of $u''$, arc $(v, w)$ enters $\mathit{loop}(u'')$, and among all such $u''$, $\mathit{loop}(u'')$ is largest.  Vertex $u''$ is well-defined since $u'$ is a candidate.  Also $u'' \not= u$, so $h(u'')$ is defined.  Replacing the part of $P$ from $s$ to $u'$ by the path in $T$ from $s$ to $p(h(u''))$ followed by the new arc $(p(h(u'')), u'')$ followed by the path in $T$ from $u''$ to $u'$ gives the desired path.
We conclude that $x$ does not dominate $y$ in $G'$.

Conversely, suppose $x$ does not dominate $y$ in $G'$.  Let $P$ be a simple path in $G'$ from $s$ to $y$ that avoids $x$.  For each new arc on $P$, if any, we replace part or all of $P$ containing the new arc by a path in $G$ that avoids $x$.  After doing all such replacements, the result is a path in $G$ from $s$ to $y$ that avoids $x$.

Suppose $P$ contains a new arc $(p(h(u)), u)$ such that $\mathit{loop}(u)$ and $\mathit{loop}(h(u))$ have a
common non-head entry $(v, w)$. If $x$ is not a descendant of $p(h(u))$ and an ancestor of $u$, then we replace $(p(h(u)), u)$ by the path in $T$ from $p(h(u))$ to $u$.  If $x$ is a descendant of $h(u)$ and a proper ancestor of $u$, we replace the part of $P$ from $s$ to $u$ by the path in $T$ from $s$ to $v$, followed by $(v, w)$, followed by a path in $\mathit{loop}(u)$ from $w$ to $u$.  Since $v$ is not in $\mathit{loop}(h(u))$, $v$ is not a descendant of $h(u)$, but $x$ is, so the path in $T$ from $s$ to $v$ avoids $x$.  Since $x$ is a proper ancestor of $u$, $x \not \in \mathit{loop}(u)$, so the path in $\mathit{loop}(u)$ from $v$ to $u$ also avoids $x$.

Suppose $P$ contains a new arc $(v, u)$ such that $\mathit{loop}(u)$ is the largest loop with non-head entry $(v, w)$. If $x$ is not an ancestor of $u$, we replace the part of $P$ from $s$ to $u$ by the path in $T$ from $s$ to $u$.  If $x$ is an ancestor of $u$ it must be a proper ancestor of $u$.  In this case we replace the arc $(v, u)$ by $(v, w)$ followed by a path in $\mathit{loop}(u')$ from $v$ to $u'$.
We conclude that $x$ does not dominate $y$ in $G$.
\end{proof}

One can compute $G'$ be augmenting the computation of $H$ appropriately.  This takes only simple data structures and $O(m)$ additional time \cite{loops:Ramalingam}.  Applying Algorithm AD to $G'$ produces the dominator tree of $G$.  Our version of Ramalingam's transformation adds only arcs that are added by Version 1 of Algorithm GD.  As we saw in Section \ref{sec:general}, the dynamics of that algorithm allow it to keep track only of arcs formed by contractions, not the extra arcs added to preserve dominators.

\section{Remarks}
\label{sec:final}

As mentioned in the introduction, Gabow~\cite{dominators:poset} has also developed an algorithm for finding dominators that uses only simple data structures and static-tree disjoint set union. The heart of his algorithm is the computation of the minimal-edge poset~\cite{poset:Gabow91,minset:poset}. This computation operates on a transformed graph.  The transformation increases the graph size by a constant factor, which increase its running time and storage space. Our algorithm operates entirely on the original graph. Gabow's presentation relies on his work on the minimal-edge poset; ours is self-contained.  His algorithm uses an nca computation on a static tree. It also uses an algorithm for finding strongly connected components, augmented to contract components as they are found.  Our algorithm does not require an algorithm for strongly connected components, and the variant presented in Section \ref{sec:loops} does not need to compute nca's.

We believe that the techniques of Buchsbaum et al. \cite{dominators:bgkrtw} can be applied to our algorithm to make it run on a pointer machine in $O(m)$ time, but the details are not straightforward, so we leave this as an open problem.  Our algorithm is simple enough that it may be competitive in practice with the LT algorithm, and we plan to do experiments to find out whether this is true and to determine the best possible practical implementation.

\bibliographystyle{plain}
\bibliography{domsurvey}

\end{document}